\documentclass[12pt]{article}
\usepackage{graphicx} 
\usepackage{amsfonts,amsmath,amssymb,amsthm}
\usepackage{mathtools}
\usepackage{subfig}
\usepackage{bbm}
\usepackage[margin=1.25in]{geometry}
\usepackage{hyperref,xcolor}
\usepackage{physics}
\hypersetup{
colorlinks,
linkcolor={red!50!black},
citecolor={blue!50!black},
urlcolor={blue!80!black}
}
\usepackage[ruled,vlined]{algorithm2e}
\usepackage{natbib}
\usepackage{setspace}
\usepackage{tikz}
\usetikzlibrary{matrix,positioning,backgrounds}
\usetikzlibrary{decorations.pathreplacing}
\usetikzlibrary{patterns}
\usepackage{multirow}

\newcounter{tprocedure}
\renewcommand{\thetprocedure}{\arabic{tprocedure}}
\newenvironment{testprocedure}[1][]{%
  \refstepcounter{tprocedure}\par\medskip
  \noindent\textbf{Procedure~\thetprocedure: #1.} \begin{itshape}%
}{%
  \medskip\end{itshape}%
}

\newcommand{\eps}{\varepsilon}

\newcommand{\R}{\mathbb{R}}
\newcommand{\E}{\mathbb{E}}
\newcommand{\cN}{\mathcal{N}}
\newcommand{\diag}{\mathrm{diag}}
\newcommand{\indep}{\perp \!\!\! \perp}
\newcommand{\rodds}{r_{\mathrm{odds}}}
\renewcommand{\P}{\mathbb{P}}
\def\ddefloop#1{\ifx\ddefloop#1\else\ddef{#1}\expandafter\ddefloop\fi}
\def\ddef#1{\expandafter\def\csname b#1\endcsname{\ensuremath{\mathbf{#1}}}}
\ddefloop ABCDEFGHIJKLMNOPQRSTUVWXYZ\ddefloop
\def\ddef#1{\expandafter\def\csname c#1\endcsname{\ensuremath{\mathcal{#1}}}}
\ddefloop ABCDEFGHIJKLMNOPQRSTUVWXYZ\ddefloop
\def\ddef#1{\expandafter\def\csname t#1\endcsname{\ensuremath{\widetilde{#1}}}}
\ddefloop ABCDEFGHIJKLMNOPQRSTUVWXYZ\ddefloop
\def\ddef#1{\expandafter\def\csname tb#1\endcsname{\ensuremath{\widetilde{\mathbf{#1}}}}}
\ddefloop ABCDEFGHIJKLMNOPQRSTUVWXYZ\ddefloop

\newtheorem{remark}{Remark}
\newtheorem{theorem}{Theorem}
\newtheorem{lemma}{Lemma}
\newtheorem{assumption}{Assumption}
\newtheorem{definition}{Definition}
\newtheorem{proposition}{Proposition}

\newcommand{\XX}{\mathbf{X}}
\newcommand{\DD}{\mathbf{D}}
\newcommand{\yy}{\mathbf{y}}
\newcommand{\ee}{\boldsymbol{\varepsilon}}
\newcommand{\GG}{\mathcal{G}}
\newcommand{\FOB}{\mathcal{F}}

\newcommand{\ER}{Erd\"{o}s–R\'{e}nyi}
\usepackage{enumerate}
\date{}

\title{\vspace{-30px}
 Finite-sample Valid Permutation Inference under Multi-way Clustering}
\author{Wenxuan Guo\thanks{Booth School of Business, University of Chicago. Email: wxguo@chicagobooth.edu}, Panos Toulis\thanks{Booth School of Business, University of Chicago. Email: ptoulis@chicagobooth.edu}, Yuhao Wang\thanks{IIIS, Tsinghua University and Shanghai Qi Zhi Institute. Email: yuhaow@tsinghua.edu.cn}}
\begin{document}

\maketitle
\onehalfspacing
\begin{abstract}
Econometric applications with multi-way clustering often feature a small number of effective clusters or heavy-tailed data, making standard cluster-robust and bootstrap inference unreliable in finite samples. In this paper, we develop a framework for finite-sample valid permutation inference in linear regression with multi-way clustering under an assumption of conditional exchangeability of the errors. Our assumption is closely related to the notion of separate exchangeability studied in earlier work, but can be more realistic in many economic settings as it imposes minimal restrictions on the covariate distribution. We construct permutation tests of significance that are valid in finite samples and establish theoretical power guarantees, in contrast to existing methods that are justified only asymptotically. While our main presentation is under exchangeable errors, our framework encompasses invariance-based tests that remain valid under general forms of heteroskedastic errors. We also extend our methodology to settings with missing data and derive power results that reveal phase transitions in detectability. Through simulation studies, we demonstrate that the proposed tests maintain correct size and competitive power, while standard cluster-robust and bootstrap procedures can exhibit substantial size distortions.
\end{abstract}

\section{Introduction}\label{sec:intro}
Cluster-robust standard errors are ubiquitous in applied econometrics and statistics whenever data exhibit within-cluster dependence. In many empirical settings, clustering occurs along more than one dimension, such as importer–exporter country pairs in international trade. This leads to dyadic (or multi-way) regression specifications, where standard one-way cluster-robust standard errors~\citep{liang1986longitudinal, arellano1987computing, moulton1986random} are no longer sufficient. 

To address this issue,~\citet{cameron2011robust} proposed a multi-way extension of standard cluster-robust variance estimators that is relatively simple to implement and has since become standard practice. However, as pointed out by~\citet{mackinnon2021wild}, the multi-way variance formulas of~\citet{cameron2011robust} are largely heuristic: they lack formal proofs of validity and can even yield non-positive definite variance matrices. Such numerical issues tend to arise precisely in settings where robust inference is needed, such as when clusters are few, small, or unbalanced.

Building on this literature,~\citet{mackinnon2021wild} provided formal justifications for multi-way cluster-robust inference. Their analysis follows the theoretical framework of {\em separately exchangeable arrays} developed by~\citet{davezies2018asymptotic, menzel2021bootstrap}, which treats the data as drawn from a two-dimensional array. 
One clustering dimension corresponds to the rows and the other to the columns of this array, with each “cell” containing the data from the corresponding cluster pair. 
Within this setup,~\citet{mackinnon2021wild} proved the asymptotic validity of $t$-statistics based on multi-way cluster-robust standard errors and proposed a wild cluster bootstrap procedure that performs well in simulations.

While elegant, the analysis of~\citet{mackinnon2021wild} relies on several assumptions inherited from the separately exchangeable framework: (i) the number of rows and columns must grow without bound, (ii) covariates and errors are jointly exchangeable in the two dimensions, (iii) data are independent across distinct rows and across distinct columns, and (iv) cell sizes are positive in expectation. These conditions are analytically convenient but can exclude many empirically relevant settings. For instance, in international trade data, importer–exporter pairs are finite, exhibit geographic dependence, and covariates are typically nonexchangeable. Such features are common in real data and violate the assumptions in earlier work.

In this paper, we develop an alternative framework for inference under multi-way clustering that replaces all four assumptions above with a single assumption of conditional separate exchangeability of the errors given covariates. Under this symmetry, we construct a permutation test of significance that is valid in finite samples and comes with theoretical power guarantees under plausible alternatives. To our knowledge, this is the first test to deliver finite-sample validity for regression models with multi-way clustered errors. The key restriction, relative to existing work, is that our exchangeability condition applies to the finite sample rather than to the population. See Section \ref{sec:overview} for a detailed comparison.

Owing to its construction, our permutation test is conceptually simple and straightforward to implement, especially when compared to multi-way wild cluster bootstrap procedures. Our tests also accommodate a finite number of clusters and impose no assumptions on the covariate distribution, thus allowing for irregular and heavy-tailed data. The extension to missing data can be handled naturally within our framework, with the permutation test performed on a carefully chosen subset of the data consisting of fully observed cells---that is, a clique in graph-theoretic terms. This novel connection to graph theory facilitates the analysis of power properties under \ER-type missingness.

In summary, a key contribution of our paper is to show that a small modification of the separate exchangeability condition in earlier work can lead to finite-sample valid inference via permutation tests and can provide a potentially valuable addition to the standard econometric toolkit.

\section{Main Idea and Overview of Results}\label{sec:overview}
To illustrate the main idea, consider the following 
regression model:
\begin{equation}\label{eq:dyadic}
y_{i j}= x_{ij}^\top\gamma +  d_{i j}^\top\beta + \eps_{ij}\;, \quad i, j = 1, \dots, n\;,
\end{equation}
where $y_{ij}\in\mathbb{R}$ are the outcomes, $d_{ij}\in\mathbb{R}^d$ are the covariates of interest, and $x_{ij}\in\mathbb{R}^p$ are auxiliary covariates; $\eps_{ij}$ are the unobserved errors. The parameter $\beta$ is of primary interest, while $\gamma$ is a nuisance parameter.
Equation~\eqref{eq:dyadic} describes a dyadic regression specification \citep{graham2020dyadic}, such as trade data between importer country $i$ and exporter country $j$. 
In the array representation of~\eqref{eq:dyadic}, index $i$ is the ``row cluster", $j$ is the ``column cluster", and 
there is a single observation per cell $(i,j)$. We will consider extensions in Section~\ref{sec:multicluster}.

Our goal is to test the null hypothesis $H_0: \beta = 0$ and perform inference via test inversion. Throughout the paper, our main assumption is the conditional exchangeability of the errors:
\begin{equation}\label{eq:ex1}
(\eps_{i j})_{i,j=1}^n \stackrel{\mathrm{d}}{=} (\eps_{\pi(i) \sigma(j)})_{i,j=1}^n \mid \{x_{ij}, d_{ij}\}_{i,j=1}^n~,~\text{for all permutations $\pi, \sigma$ on  $\{1, \dots, n\}$.}
\end{equation}
We extend our framework to accommodate heteroskedastic errors in Section \ref{sec:het}.
By Equation \eqref{eq:ex1}, the error disturbances have an importer--exporter ``random-effects'' structure conditional on the covariates. Relative to earlier work, this assumption is weaker in that it imposes no restrictions on the covariate distribution, but stronger in that it holds in finite samples rather than the population. We give a detailed comparison in the following section.

Under this invariance assumption, we propose a finite-sample valid permutation test for $H_0$. 
The procedure, detailed in Section~\ref{sec:procedure}, consists of the following three steps:
\begin{enumerate}
    \item ({\bf Partialling-out.}) We construct a group $\mathcal{G}$ of two-way permutations and matrices $V_g$ for each $(\pi, \sigma) \coloneqq g\in\mathcal{G}$ such that
\begin{equation}\label{eq:Vg}
V_g^\top \XX = 0\;,~\text{and}~V_g^\top \XX^g = 0\;,    
\end{equation}
where $\XX=[x_{ij}]\in\mathbb{R}^{n^2 \times p}$ are the auxiliary covariates in matrix form, and $\XX^g$ is obtained from $\XX$ after permuting the rows according to $\pi$ and 
the columns according to $\sigma$. Similarly, define $\yy$ and $\ee$ as outcome and error vectors in $\R^{n^2}$, respectively, following the same stacking operation as for $\XX$.

This step constructs the matrices $\{V_g\}_{g\in\GG}$, which act as orthonormal projections to partial out nuisance effects, similar to the Frisch–Waugh–Lovell theorem \citep{frisch1933partial, lovell1963seasonal}. Section~\ref{sec:procedure} provides simple conditions under which such matrices exist.

\item ({\bf Minorized statistics.}) We consider test statistics $s(\yy,V)$ that depend on $\yy$ only through $V \yy$. Following from Equation \eqref{eq:Vg} and the linear model~\eqref{eq:dyadic}, we can thus write $s(\yy^g,V_g) = f(V_g \ee^g)$ under the null hypothesis, for some known function $f$. Next, define the statistic's {\em minorized} version as $s^*(\yy) =  \min_{g\in\mathcal{G}} s(\yy, V_g)$. 
Then, under the null hypothesis, we have
\begin{equation}\label{eq:minorized}
    s^*(\yy) = \min_{g\in\mathcal{G}} f(V_g \ee) := f^*(\ee)\;.
\end{equation}
Thus, minorization provides an explicit function of the unobserved errors.

\item ({\bf Permutation test.}) We compare $s^*(\yy)$ with the values 
 $\{ s(\yy^g, V_g)\}_{g\in\mathcal{G}}$, i.e., the randomization 
distribution constructed from two-way permutations $\mathcal{G}$ on the outcomes. 
For instance, we may define the $p$-value through the function
$$
\mathrm{pval}(s^*(\yy), \{ s(\yy^g, V_g)\}_{g\in\mathcal{G}}) = 
\frac{1}{1+|\mathcal{G}|}\Bigl(1 + \sum_{g\in \mathcal{G}}\mathbbm{1}\{ s^*(\yy) \le s(\yy^g, V_g)\}\Bigr)\;.
$$
\end{enumerate}

It is straightforward to verify that the resulting testing procedure is valid in finite samples for any $n>0$ since
    \begin{align}\label{eq:proof_concept}
\underbrace{\mathrm{pval}(s^*(\yy), \{ s(\yy^g, V_g)\}_{g\in\mathcal{G}})}_{\text{feasible test}} & = \mathrm{pval}(f^*(\ee), \{ s(\yy^g, V_g)\}_{g\in\mathcal{G}}) \nonumber\\
& =
\mathrm{pval}(f^*(\ee), \{ f(V_g\ee^g)\}_{g\in\mathcal{G}})\;,\nonumber\\
& \ge
\underbrace{\mathrm{pval}(f^*(\ee), \{ f^*(\ee^g)\}_{g\in\mathcal{G}})}_{\text{finite-sample valid, infeasible}}\;.
\end{align}

In the above derivation, the first equality follows from~\eqref{eq:minorized}. The second step follows from the null hypothesis, as argued above.
The last step follows from the minorization property. 
Note that the last $p$-value is finite-sample exact due to the invariant hypothesis~\eqref{eq:ex1}. However, such a $p$-value is infeasible because it depends on 
transformations of the errors, $\ee^g$, which are unobserved. 
The minorization step yields a feasible test that remains finite-sample valid, albeit slightly conservative due to \eqref{eq:proof_concept}; we quantify this conservativeness through simulations.
The full proof of validity is given in Theorem~\ref{thm:valid}.

An important feature of our construction is that the derivation in~\eqref{eq:proof_concept} is not specific to two-way permutations and applies under any invariance group $\mathcal{G}$, provided that an invariance assumption like~\eqref{eq:ex1} holds. We exploit this flexibility to handle error heteroskedasticity, missing data, and multi-way clustering. 
See Sections~ \ref{sec:het}, \ref{sec:missing}, and~\ref{sec:multicluster}, respectively.

 Beyond type I error control, in Theorem~\ref{thm:power} we derive a power theory for our test, which is asymptotic and characterizes the minimal signal strength of $|\beta|$ required to ensure a diminishing type II error. This result depends explicitly on a parameter of the tail behavior of regression errors and can be adapted to missing data~(Section \ref{sec:missing}). 
 This new test preserves finite-sample validity if two-way exchangeability holds conditional on the pattern of missingness~(Theorem \ref{thm:cond_valid}). To maintain high power, we propose an efficient block-construction algorithm based on graph methods. Using this construction, we derive theoretical power bounds under \ER-type missingness. This analysis quantifies how the density of cells with observed data as compared to all possible cells in the data array affects treatment effect detection~(Theorem~\ref{thm:power_miss}). Finally, in Section \ref{sec:multicluster}, we generalize our procedure to multi-way clustering, including classical two-way layouts and panel data.

\subsection{Related Work}
Our work contributes a new perspective to a long line of research on cluster-robust standard errors, developed mainly in econometrics~\citep{liang1986longitudinal, arellano1987computing, bertrand2004much, kezdi2004robust, hansen2007asymptotic}. In subsequent work,~\citet{cameron2011robust} extended these methods to multi-way clustering. A rigorous analysis of the two-way clustering regime was initiated by~\citet{davezies2018asymptotic, menzel2021bootstrap}, based on the notion of separate exchangeability of array-valued data.

In a significant development,~\citet{mackinnon2021wild} built upon this theory to establish the asymptotic properties of $t$-statistics based on two-way cluster-robust standard errors; see also~\citet{djogbenou2019asymptotic, conley2018inference}.
Their analysis relies on four key assumptions: (i) both cluster dimensions grow without bound, (ii) covariates and error disturbances are jointly exchangeable, (iii) data are independent across distinct cluster pairs, and (iv) the number of observations within each cluster pair is positive in expectation.
Assumption~(ii), in particular, entails an invariance condition that can be written as
\begin{equation}\label{eq:inv_mackinnon}
(d_{ij}, x_{ij}, \eps_{ij})_{i,j=1}^n \stackrel{\mathrm{d}}{=} (d_{\pi(i)\sigma(j)}, x_{\pi(i)\sigma(j)}, \eps_{\pi(i)\sigma(j)})_{i,j=1}^n\;,
\end{equation}
for all permutations $\pi, \sigma$ on $\{1, \dots, n\}$.
This condition is closely related to our invariance assumption~\eqref{eq:ex1}, as both impose a form of two-way (double) exchangeability on the data.

However, this comparison also underscores key conceptual differences between our framework and the separate exchangeability approach adopted in earlier work.
Specifically, unlike assumption~\eqref{eq:inv_mackinnon}, our invariance assumption~\eqref{eq:ex1} imposes no restrictions on the covariate distribution.
Our approach is therefore more realistic in empirical econometric settings where the covariates differ systematically across clusters and cannot reasonably be modeled as random effects.
For instance, in international trade data, $x_{ij}$ may represent the distance between countries $i$ and $j$.
In this context, it is not realistic to assume that distance is exchangeable with respect to $i$ or $j$, and it is even unclear what a superpopulation distribution of between-country distances would really represent.
By contrast, our approach in this paper is to conduct inference conditional on the observed $x_{ij}$.

At the same time, our invariance assumption~\eqref{eq:ex1} is required to hold for any fixed covariates in the sample, whereas the restriction in~\eqref{eq:inv_mackinnon} is on the hypothetical superpopulation. For example, suppose $x_{ij} = f(\eta_{ij})$ and $\eps_{ij} = g(\eta_{ij}, u_{ij})$ for some functions $f, g$ and doubly exchangeable arrays $\eta_{ij}, u_{ij}$.
In this case, our assumption could be violated because $\eps_{ij}$ may no longer be exchangeable conditional on $f(\eta_{ij})$, whereas condition~\eqref{eq:inv_mackinnon} would continue to hold assuming the superpopulation is well defined.
The limitation of assuming in-sample invariance is that it excludes heteroskedastic errors that may depend on the covariates; in Section~\ref{sec:het}, we discuss an extension of our test that continues to deliver finite-sample valid $p$-values under error heteroskedasticity.
The benefit of this assumption is that it enables the construction of permutation tests that are valid in finite samples, while dispensing with all other assumptions (i)–(iv) outlined above.
In short, the distinction arises because our assumption corresponds to a “fixed-design” invariance, whereas earlier work imposes a “random-design” invariance.
Thus, neither framework is a special case of the other, and from that perspective they are best viewed as complementary.

In terms of methodology, our paper is more closely related to an emerging literature in statistics that studies testing and inference under data invariance assumptions. Building on the seminal work of~\citet{freedman1983nonstochastic},~\citet{toulis2019life, toulis2025biometrika} formalized early ideas in invariance-based inference through the construction of test statistics that, under the null hypothesis, become explicit functions of the errors, as in Equation~\eqref{eq:minorized}. The two-way exchangeability invariance~\eqref{eq:ex1} that we use in this paper was originally introduced by~\citet[Section~4.4]{toulis2019life}. However, the randomization tests developed in that paper depend on OLS residuals and are generally not finite-sample valid, except in special cases. An important advance toward this goal was achieved by~\citet{Lei2020} and later substantially extended by~\citet{wen2025residual}. Our construction based on partialling-out and minorization, as described earlier, and the theoretical framework we use to study power follow closely~\citet{wen2025residual}. However, those papers rely on simple one-dimensional exchangeability and do not account for clustering in the regression errors. In addition, they do not address the problem of missing data. 

 In the social sciences and psychology,~\citet{mantel1967detection, krackhardt1988predicting} proposed permutation tests for dyadic regression---a special case of two-way clustering---known as quadratic assignment procedures (QAP). However, these tests are not finite-sample valid for the general regressions considered in our paper; see also~\citet{dekker2007sensitivity, shi2025qap}. 
More broadly, permutation tests frequently arise in experimental and causal inference settings as Fisherian randomization tests~\citep{fisher1935design}. These procedures are finite-sample valid under sharp null hypotheses, even for complex test statistics~\citep{gerber2012field, guo2025mlfrt}, but generally require additional modifications to maintain validity under weaker null hypotheses (e.g., regression tests of significance), at least in an asymptotic sense~\citep{romano1990behavior, chung2013exact, wu2021randomization, ding2016variation, basse2019randomization, puelz2021graph}. A notable exception is the use of the minorization trick to construct finite-sample valid tests of certain weak null hypotheses~\citep{caughey2023randomisation}, which also serves as an inspiration for our work. In settings with missing data, we follow a similar approach to~\citet{puelz2021graph} by transforming the problem of constructing block permutations into a biclique search problem in graph theory. 
%

\section{Main Test under Dyadic Regression}\label{sec:procedure}
In this section, we remain focused on the dyadic regression setting~\eqref{eq:dyadic} with $n$ row clusters, $n$ column clusters, and each cell $(i,j)$ containing a single observation. We will generalize this setup to multi-way clustering in Section~\ref{sec:multicluster}.
For notational simplicity, we express the dyadic regression model \eqref{eq:dyadic} in matrix form,
\begin{equation}\label{eq:model_matrix}
    \yy = \XX \gamma + \DD \beta + \ee\;.
\end{equation}
Here, $\yy$, $\XX$, $\DD$, and $\ee$ collect all dyadic observations into long matrices (or vectors) through a stacking operation. Concretely, we stack the dyadic entries in lexicographic order, and create vectors when stacking scalar values and matrices when stacking vectors; i.e.,
\begin{align*}
    \yy &= (y_{11}, \dots, y_{1n}, y_{21}, \dots, y_{2n}, \dots, y_{n1}, \dots, y_{nn})^\top\in\R^{n^2}\;,\\
    \XX &= (x_{11}, \dots, x_{1n}, x_{21}, \dots, x_{2n}, \dots, x_{n1}, \dots, x_{nn})^\top\in\R^{n^2\times p}\;.
\end{align*}
That is, the observation associated with the cell $(i,j)$ occupies row $(i-1)n + j$ in the stacked objects. In this notation, the error vector $\ee$ is defined similarly as $\yy$. $\DD\in\R^{n^2 \times d}$ is a long matrix when $d>1$ and reduces to a vector when $d = 1$. Our procedure below will handle both cases.

Our goal is to test $H_0: \beta = 0$ under the key assumption of double exchangeability discussed in the introduction. We repeat the definition here for completeness.
\begin{assumption}\label{asmp:double_ex}
For any permutations $\pi, \sigma$ on $[n]:=\{1,2, \ldots, n\}$, the errors satisfy 
\begin{equation*}
(\eps_{i j})_{i,j\in[n]} \stackrel{\mathrm{d}}{=} (\eps_{\pi(i) \sigma(j)})_{i,j\in[n]} \mid \XX, \DD\;.
\end{equation*}
\end{assumption}
By Assumption \ref{asmp:double_ex}, the $n\times n$ matrix of dyadic errors is both row-wise and column-wise exchangeable. Double exchangeability holds, for instance, in ``random-effects" models,
\begin{equation}\label{eq:random_eff1}
\eps_{i j}=\eta_i+\xi_j+ u_{i j}~,    
\end{equation}
where $\eta_i, \xi_j, u_{i j}$ are i.i.d. random variables conditional on covariates. 
Then, Assumption \ref{asmp:double_ex} holds since, conditional on covariates,
$$
\eps_{\pi(i) \sigma(j)}=\eta_{\pi(i)}+\xi_{\sigma(j)}+u_{\pi(i)\sigma(j)} \stackrel{d}{=} \eta_{i}+\xi_{j}+u_{ij} = \eps_{ij}\;.
$$
Note that the assumption does not require that $\eps_{ij}$ are independent of the covariates $\XX,\DD$. For instance, the assumption allows $\xi_{j} = f(\XX,\DD) \tilde \xi_{j} + g(\XX,\DD)$, where $\{\tilde\xi_{j}\}_{j\in[n]}$ represent generic i.i.d. random variables, and $f, g$ are arbitrary functions. On the other hand, Assumption \ref{asmp:double_ex} excludes heteroskedastic errors whose distributions depend on the covariates; we study an extension of our test that addresses this issue in Section \ref{sec:het}.

We also note that under Assumption~\ref{asmp:double_ex}, we may further assume $\E(\eps_{ij} | x_{ij}, d_{ij})=0$ without loss of generality, since the invariance is maintained under additive shifts.
As a result, parameters $\beta$ and $\gamma$ retain their standard regression interpretation.

\subsection{Main Testing Procedure}
Our main procedure now follows the three steps outlined in Section~\ref{sec:overview}.
First, we randomly generate $K$ row permutations $\pi_{1},\dots,\pi_{K}$ and $K$ column permutations $\sigma_{1},\dots,\sigma_{K}$, and define the full set of two-way permutations as 
\begin{equation}\label{eq:pi_K}
\GG = \left\{\left(\pi_{0}, \sigma_{0}\right),\left(\pi_{1}, \sigma_{1}\right), \ldots,\left(\pi_{K}, \sigma_{K}\right)\right\}\;,
\end{equation}
where $\pi_{0} = \sigma_{0} = \mathrm{Id}$ is the identity mapping on $[n]$. For now, we assume that $\GG$ forms an algebraic group in both dimensions. That is, for any $r,s \in\{0,1, \ldots, K\}$, there exists $k\in\{0, \dots, K\}$ such that $\pi_{k}=\pi_{r} \circ \pi_{s}$, $\sigma_{k}=\sigma_{r} \circ \sigma_{s}$, where $\circ$ denotes function composition.
In the following section, we will provide a specific algorithmic construction of $\GG$, which will guarantee the group property.

Given a row permutation $\pi$ and a column permutation $\sigma$, we can apply them to 
outcomes $\yy$ by replacing $y_{ij}$ with $y_{\pi(i)\sigma(j)}$. Hence, we define the permuted data by $\yy_{\pi, \sigma}$, where the $[(i-1)n+j]$-th entry of $\yy_{\pi, \sigma}$ is equal to $y_{\pi(i) \sigma(j)}$, and define the permuted covariates $\XX_{\pi, \sigma}$ similarly. Finally, let $N = n^2$ denote the total number of observations in this setup. We are now ready to describe our main testing procedure.

\begin{testprocedure}[Invariant Permutation Test under Dyadic Regression]\label{proc1}
\begin{enumerate}
    \setlength{\itemsep}{-1.2em}
    \item For $k = 1,\dots, K$, find an orthonormal matrix $V_k\in\R^{N\times (N-2p)}$ such that $V_k^\top \XX=0$ 
    and $V_k^\top \XX_{\pi_{k},\sigma_{k}}=0$. Compute 
    \[
    a_k = \|\DD^\top V_k V_k^\top \yy\|\;, \quad b_k = \| \DD^\top V_k V_k^\top \yy_{\pi_{k},\sigma_{k}}\|\;.
    \]
    For $d = 1$, using $\langle \cdot, \cdot \rangle$ to denote the inner product, the quantities above reduce to
    \[
    a_k = |\langle V_k^\top \DD,   V_k^\top \yy\rangle|\;, \quad b_k = |\langle V_k^\top \DD, V_k^\top \yy_{\pi_{k},\sigma_{k}} \rangle|\;.
    \]
    \item Compute the randomization $p$-value
    \begin{equation}\label{eq:rand_pval}
        \mathrm{pval} \coloneqq \frac{1}{K+1}\Bigl(1+\sum_{k=1}^K \mathbbm{1}\bigl\{\min _{1 \leq j \leq K} a_j \leq b_k\bigr\}\Bigr)\;.
    \end{equation}
    \item \textbf{Inference.} Let $\mathrm{pval}(b)$ denote the randomization $p$-value for $H_0^b: \beta = b$, and define the confidence region
    \begin{equation*}
        \mathrm{CI} = \{x \in \R^d: \mathrm{pval}(x) > \alpha \}\;.
    \end{equation*}
    Then, $\mathrm{CI}$ is the $100\times (1-\alpha) \%$ confidence region for $\beta$.
\end{enumerate}
\end{testprocedure}

As described in the introduction, Step 1 of the procedure constructs an orthonormal matrix $V_k$ that projects $\yy$ onto the orthogonal complement of the nuisance-parameter space. 
The projection matrix $V_kV_k^\top$ can be computed explicitly as $V_kV_k^\top = I_N-\XX_k (\XX_k^\top\XX_k)^+ \XX_k^\top$, where $\XX_k=\bigl(\XX,\XX_{\pi_k,\sigma_k}\bigr)\in \R^{N\times 2p}$ and $A^+$ denotes the Moore-Penrose inverse of $A$. This removes nuisance effects from $\XX$ (partialling-out step). The other key ingredient is the minorization trick in~Equation~\eqref{eq:rand_pval} of Step 2, which guarantees 
that the reference distribution of the test stochastically dominates the true reference distribution of the idealized test. 
The resulting $p$-value is therefore finite-sample valid conditionally on the data according to the following theorem, and thus it is valid marginally as well. See Section \ref{sec:proof_IPT} in the Appendix for the proof.
\begin{theorem}\label{thm:valid}
Let $(\XX, \DD, \yy)$ be the data from model \eqref{eq:model_matrix} with $p< N/2$ and suppose that Assumption~\ref{asmp:double_ex} holds. Under $H_0: \beta = 0$, the $p$-value defined in \eqref{eq:rand_pval} satisfies
\begin{equation*}
    \mathbb{P}(\mathrm{pval} \leq \alpha\mid \XX, \DD) \leq \alpha~,~\text{for all}~\alpha \in[0,1],~n > 0\;.
\end{equation*}
\end{theorem}
To our best knowledge, Procedure \ref{proc1} delivers the first general finite-sample valid test of significance in dyadic regression. Under Assumption \ref{asmp:double_ex}, it controls the type I error at the nominal level for any error and covariate distribution, allowing heavy-tailed distributions and strong cluster-level dependence. 
We demonstrate these points through simulations in Section \ref{sec:simu} and the empirical application in Section \ref{sec:trade}.

\begin{remark}\label{rmk:random}
We note that Procedure \ref{proc1} requires a permutation group $\GG$ that is random, as described in the following section. Other tests proposed in this paper also involve certain random components. As a result, running our procedures with different random seeds may yield slightly different $p$-values. A practical remedy is to compute multiple $p$-values and take their median, as suggested by \citet{diciccio2020exact, guo2025rank}. We will adopt this strategy for the test applied in Section \ref{sec:trust} of the Appendix.
\end{remark}

\begin{remark}\label{rmk:conservative}
Due to the minorization step in Equation \eqref{eq:rand_pval}, our method can be conservative, with the actual Type I error rate below the nominal level $\alpha$. The simulation results in Section \ref{sec:simu} are consistent with this observation. To reduce this conservativeness, one could adapt the pairwise randomization test proposed in \cite{guan2024conformal, zhong2024unconditional}, which has been shown to exhibit improved empirical performance. We discuss this in more detail in Section \ref{sec:extension}, with supporting simulation results in Section~\ref{sec:simu}.
\end{remark}

\begin{remark}\label{rmk:missing}
The finite-sample valid test above applies to fully observed dyadic data. In the presence of missing entries, we extend our approach by constructing a conditional permutation test based on fully observed blocks, within which double exchangeability is preserved. We propose a graph-theoretic algorithm based on biclique search to identify such blocks. Further details are provided in Section \ref{sec:missing} of the Appendix.
\end{remark}

\subsection{Construction of Permutation Subgroup}
In this section, we propose an algorithm as a concrete way to construct $\GG$ that has the appropriate group structure required for Procedure~\ref{proc1}. The concrete algorithm is adapted from Algorithm 1 of \citet{wen2025residual} to the dyadic regression setting, and is presented in Algorithm~\ref{alg:construct} that follows.
\begin{algorithm}
	\DontPrintSemicolon
	\KwData{The index set $I$ to permute and the number of permutations $K$.}
	\Begin{
        $n \gets |I|$\;
        Generate a random one-to-one mapping $\pi: I \to [n]$\;
        \For{$k=1, \dots, K$}{
        Construct a permutation function $\psi_k:=\pi^{-1} \circ \tilde{\psi}_k \circ \pi$, where $\tilde{\psi}_k$ is a permutation function such that
        \begin{enumerate}
            \item If $i > (K+1) \lfloor\frac{n}{K+1}\rfloor$, set $\tilde{\psi}_k(i) = i$, 
            \item If $i \le (K+1) \lfloor\frac{n}{K+1}\rfloor$, set
        \end{enumerate}
        $$
        \tilde{\psi}_k(i):= \begin{cases}i+k & ~\text{if}~(i-1) \bmod (K+1) < K+1-k \\ i-(K+1-k) & ~\text{if}~ (i-1) \bmod (K+1) \ge K+1-k\end{cases}\;.
        $$
        }
    }
    \KwResult{Set of permutations $\psi_k$, $k = 0, \dots, K$, where $\psi_0 = \mathrm{Id}$ is the identity mapping on $I$.}
	\caption{Permutation set construction \label{alg:construct}}
\end{algorithm}

To understand the algorithm, consider $I = [n]$, $\pi = \mathrm{Id}$, and suppose $n$ can be divided by $K+1$. Then, our algorithm divides the $n$ indices into $m\coloneqq n /(K+1)$ ordered lists of indices and performs cyclic permutations on each sub-list. Specifically, we first denote $S_1, \ldots, S_m$ as $m$ ordered lists of indices such that
$$
(1, \ldots, n)\coloneqq (\underbrace{1, \ldots, K+1}_{S_1}, \underbrace{K+2, \ldots, 2(K+1)}_{S_2}, \underbrace{(m-1)(K+1)+1, \ldots, m(K+1)}_{S_m})\;,
$$
Then we define $\tilde{\psi}_k$ as the permutation for which
$$
\left(\tilde{\psi}_k(1), \ldots, \tilde{\psi}_k(n)\right):=\left(S_1^k, \ldots, S_m^k\right)\;,
$$
where each $S_i^k$ is created via shifting all the elements in $S_i$ by $k$ places. Taking $S_1^k$ for example, we have $S_1^k\coloneqq (K+2-k, \ldots, K+1,1,2, \ldots, K+1-k)$. One can easily verify that the resulting set of permutations $\psi_0, \psi_1, \dots, \psi_K$ formalizes a group, since it is constructed by cyclic permutations. We give a formal proof in Section \ref{sec:proof_IPT}. For dyadic regression, it suffices to consider $I = [n]$, while in the next few sections, we will apply Algorithm \ref{alg:construct} with other choices of index sets.

For dyadic regression and Procedure~\ref{proc1}, we just need to apply Algorithm \ref{alg:construct} twice and store the outputs as $\{\pi_{k}\}_{k=0}^K$ and $\{\sigma_{k}\}_{k=0}^K$, which represent the row and column permutations, respectively. Then, one can verify that $\GG$ defined in \eqref{eq:pi_K} is an algebraic group, as its row and column permutations are both algebraic groups. 
Importantly, Algorithm \ref{alg:construct} provides meaningful permutations that lead to a provably powerful test as we show in the next section.

\subsection{Test under Heteroskedasticity}\label{sec:het}
The permutation test relies on double exchangeability, which generally rules out heteroskedasticity across dyads. To accommodate heteroskedastic errors, we consider an alternative invariance assumption based on double sign symmetry.

\begin{assumption}\label{asmp:double_sign}
For any sign vectors $s, t \in\{-1,1\}^n$, the errors satisfy
\[
(\eps_{ij})_{i,j\in[n]}
\stackrel{\mathrm{d}}{=}
(s_i t_j \eps_{ij})_{i,j\in[n]}
\mid \XX,\DD.
\]
\end{assumption}

Double sign symmetry is analogous to double exchangeability, but replaces the permutation group $\cG$ introduced in Section~\ref{sec:overview} with a sign-flip group. This provides another example of an invariance assumption \citep{guo2023invariance} that can be accommodated within our framework. Importantly, such an approach can handle a wide range of heteroskedasticity settings parameterized as
\[
\eps_{ij} = h(X_{ij}) u_i v_j,\quad u_i\stackrel{iid}{\sim} F_U, \quad v_j\stackrel{iid}{\sim} F_V,
\]
where $F_U, F_V$ are two symmetric distributions and $h$ is a possibly unknown measurable function of covariates. The main test in our paper (Procedure 1) would remain unchanged except for the application of random sign flips in Step 1 of our procedure. As we noted in Section \ref{sec:overview}, the validity argument continues to hold under any invariance, including Assumption \ref{asmp:double_sign}.  We present the concrete testing procedure in Section \ref{sec:sign} of the Appendix and illustrate its finite-sample performance through a simulation study.

\section{Power Analysis}\label{sec:power}
In this section, we focus on the $d=1$ setting (testing a single coefficient) and study the statistical power of Procedure \ref{proc1} combined with the construction from Algorithm \ref{alg:construct}. To obtain meaningful results, we impose some necessary yet mild moment conditions on dyadic errors. First, rewrite the covariate vector $\DD$ as a linear model on $\XX$:
\begin{equation*}
\DD = \XX \beta^D + e\;,
\end{equation*}
where $e \in \R^{N}$ is an error vector associated with $\DD$. Note that this representation is without loss of generality, as one can always set $e\coloneqq \DD - \XX\beta^D$ for a given vector $\beta^D$. 
Next, we impose the following assumption on $e$ and $\ee$.
\begin{assumption}\label{asmp:model}
The error vectors $\ee$ and $e$ are from a random-effects model:
\begin{align*}
    \eps_{ij} &= \eta_i+\xi_j+\zeta_{i j}\;,\quad e_{ij}= \eta^e_i+\xi^e_j+\zeta^e_{i j}\;.
\end{align*}
Here, $\eta_i\stackrel{iid}{\sim}\P_{\eta}$, $\xi_j\stackrel{iid}{\sim}\P_{\xi}$, $\zeta_{ij}\stackrel{iid}{\sim}\P_{\zeta}$ and $\eta^e_i\stackrel{iid}{\sim}\P_{\eta}^e,\xi^e_j\stackrel{iid}{\sim}\P_{\xi}^e, \zeta^e_{i j}\stackrel{iid}{\sim}\P_{\zeta}^e$, all with mean zero.
\end{assumption}

Assumption \ref{asmp:model} requires that the errors associated with $\DD$ and $\yy$ are both generated from random features, implying that the dyadic errors $\ee$ are doubly exchangeable. Additionally, we impose certain moment conditions on the distribution of random features.
\begin{assumption}\label{asmp:err}
The random features of $\{\eps_{ij}\}_{i, j\in[n]}$ in Assumption \ref{asmp:model} satisfy
\begin{align*}
    0&<\E|\eta_{1}|^{1+t} < \infty, \quad 0<\E|\xi_{1}|^{1+t} < \infty, \quad 0<\E|\zeta_{11}|^{1+t} < \infty\;, 
\end{align*}
for some constant $t\in[0, 1]$. For some $\kappa >0$, the random features of $\{e_{ij}\}_{i,j\in[n]}$ satisfy
\begin{gather*}
  0<\E|\eta^e_{1}|^2 < \infty,\quad 0<\E|\xi^e_{1}|^2 < \infty,\quad 0<\E|\zeta^e_{11}|^2 < \infty\\
  \E (\zeta^e_{11})^2 \ge \kappa (\E (\eta_1^e)^2 + \E (\xi_1^e)^2)\;.
\end{gather*}
\end{assumption}
Assumption \ref{asmp:err} requires that the random features in the error vector $e$ have finite second moments, and the idiosyncratic variance dominates the cluster-level variances by a factor $\kappa$. 
Crucially, our power analysis accommodates heavy-tailed components in dyadic errors $\ee$, and our result depends explicitly on the moment parameter $t$. Similar assumptions have been employed for power analysis of randomization tests in standard linear models \citep{wen2025residual}; here, we extend their assumption to doubly exchangeable errors and perform a refined power analysis. Concretely, the following theorem shows that the type II error of Procedure~\ref{proc1} vanishes asymptotically. Throughout the paper, $a_n = \Omega(b_n)$ means that $a_n \ge c b_n$ for some constant $c$ and all sufficiently large $n$, whereas $a_n = \omega(b_n)$ indicates $a_n / b_n \to \infty$.
\begin{theorem}\label{thm:power}
Suppose that $K$ is fixed and $n$ can be divided by $K+1$. Suppose that Assumptions \ref{asmp:model} and \ref{asmp:err} hold. Let $\mathrm{pval}$ be the randomization $p$-value from Procedure \ref{proc1} with the permutation set generated from Algorithm \ref{alg:construct}. In the asymptotic regime where $\beta$ and $p$ vary with $n$ in a way such that $n >(\max\{2/\kappa, 2\}+m)p$ for some constant $m>0$ and
$$
|\beta|=\Omega\left(n^{-\frac{t}{1+t}}\right)~\text {if}~t<1 \quad~\text {or}~ \quad|\beta|=\omega\left(n^{-\frac{1}{2}}\right)~\text{if}~t=1\;,
$$
we have $\lim_{n \to \infty} \P(\mathrm{pval}>\frac{1}{K+1}\mid \XX, \DD)=0$.
\end{theorem}
Intuitively, Theorem \ref{thm:power} characterizes the minimal signal strength $|\beta|$ to guarantee a diminishing type II error, and a smaller threshold on $|\beta|$ indicates higher statistical power. Notably, the threshold increases as the moment parameter $t$ decreases, reflecting that heavy-tailed dyadic errors (smaller $t$) make the testing problem more difficult. Note that Assumptions \ref{asmp:model} and \ref{asmp:err} are used only in power analysis, and are not required for validity.

Interestingly, the minimal signal strength in Theorem \ref{thm:power} is identical to that for linear regression with $n$ observations under simple error exchangeability~\citep[Theorem 3]{wen2025residual}. 
This implies that testing in the dyadic regression model with $n$ clusters is ---in terms of detectability--- as hard as testing the linear model with $n$ observations. This result may be viewed as an analogue of asymptotic results for cluster-robust statistics, where the convergence rate depends on the number of clusters rather than the total number of observations \citep{davezies2018asymptotic, mackinnon2021wild}.
In our analysis, this phenomenon is a result of Assumption \ref{asmp:model}, which introduces cluster-level error components $(\eta,\xi)$ and thus reduces the effective sample size from $n^2$ to $n$. Under a stronger assumption that $\eta = \xi = 0$, the dyadic structure can be better exploited to obtain strictly better signal strengths as shown below.
\begin{theorem}\label{thm:power_iid}
Suppose the assumptions in Theorem \ref{thm:power} hold and that $\eta = \xi = 0$ with probability one. In the asymptotic regime where $\beta$ and $p$ vary with $n$ in a way such that $n>(\max\{2/\kappa, 2\}+m)p$ for some constant $m>0$ and
$$
|\beta|=\Omega\left(n^{-\frac{2t}{1+t}}\right)~\text {if}~t<1 \quad~\text {or}~ \quad|\beta|=\omega\left(n^{-1}\right)~\text{if}~t=1\;,
$$
we have $\lim_{n \to \infty} \P(\mathrm{pval}>\frac{1}{K+1}\mid \XX, \DD)=0$.
\end{theorem}
Direct calculation shows that the minimal signal strength in Theorem \ref{thm:power_iid} equals the square of the corresponding threshold in Theorem \ref{thm:power}; this is because the effective sample size increases from $n$ to $n^2$. See Section \ref{sec:proof_IPT_power} for detailed proofs.

\begin{remark}\label{rmk:missing_power}
In the presence of missing data, the power of our test decreases due to the loss of observations. In Section \ref{sec:power_CRP}, we consider an \ER-type missingness, and derive a complete power theory for the conditional permutation test discussed in Remark \ref{rmk:missing}. Consider an i.i.d. missing pattern with missingness probability $1 - \rho \in (0, 1)$. We establish the minimum detectable signal strength by replacing the sample size $n$ in Theorem \ref{thm:power} with a parameter $s\le n$:
\begin{gather*}
|\beta|=\Omega\left(s^{-\frac{t}{1+t}}\right)~\text {if}~t<1 \quad~\text {or}~ \quad|\beta|=\omega\left(s^{-\frac{1}{2}}\right)~\text{if}~t=1\;,\\
s = \min\Bigl\{\rodds\log n, \frac{n^{1/3}}{\log n}\Bigr\} \;, \quad \rodds = \rho / (1-\rho)\;.
\end{gather*}
Here, $s$ is the size of the largest fully observed block (biclique) induced by the \ER-type missingness and is formally defined in Section \ref{sec:power_CRP}. To better understand these results, consider the following scenarios:
\begin{enumerate}[(i)]
    \item {\bf Constant missingness probability.}  
    When \(\rho \in (0,1)\) is fixed, the largest biclique has size of order \(\log n \). 
    This is a well-established result in random graph theory~\citep{frieze2015introduction}.  
    Consequently, our test requires exponentially larger sample sizes to achieve a vanishing type~II error.  
    That is, loosely speaking, to detect a signal of magnitude \(|\beta| = \epsilon\), we require \(e^{1/\epsilon}\) data points. 
    This illustrates the fundamental impact of missing data on efficiency loss.  

    \item {\bf Diminishing missingness probability I.}  
    Suppose that \(1-\rho \to 0\) and, in addition, \(r_{\mathrm{odds}} = O(n^{1/3}/\log^2 n)\); that is, the odds of observation increase slowly at this rate.  
    Then, the biclique size satisfies
    \[
    s \sim r_{\mathrm{odds}}\log n\;.
    \]
    In other words, \(s\) grows proportionally to the odds of observation up to a logarithmic factor.  
    Compared with the first case, \(s\) can grow polynomially in \(n\),  
    and thus the minimum detectable signal can decay at a polynomial rate, though not necessarily at the canonical \(n^{-1/2}\) rate.

    \item {\bf Diminishing missingness probability II.}  
    Suppose that \(1 - \rho \to 0\) and \(r_{\mathrm{odds}} > n^{1/3}/\log^2 n\),  
    corresponding to rapidly increasing odds of observation.  
    Then, the biclique size saturates at order \(n^{1/3}/\log n\).  
    Beyond this threshold, further increases in the odds of observation do not enlarge the biclique size.
\end{enumerate}

We note that our power analysis under missing data is novel and potentially of independent interest. We refer readers to Section \ref{sec:power_CRP} of the Appendix for a detailed discussion.
\end{remark}

\section{Generalization to Multi-way Clustering}\label{sec:multicluster}
In this section, we extend our methodology to settings with multi-way clustered errors.
The extension is straightforward within our framework because the key argument for finite-sample validity, given in Equation~\eqref{eq:proof_concept}, does not rely on the specific structure of double exchangeability.  
Accordingly, our discussion focuses on suitable modifications of the invariance property~\eqref{eq:ex1} for each new domain, while implementation follows the same structure as Procedure~\ref{proc1} and Procedure~\ref{proc:CIPT} summarized in Remark \ref{rmk:missing}.

As a starting point, consider the following three-way regression model:
\begin{equation}\label{eq:model_threeway}
    y_{ijl} = x_{ijl}^\top \gamma + d_{ijl}^\top\,\beta + \eps_{ijl}, 
    \quad i \in [m],~j \in [n],~l \in [\ell]\;.
\end{equation}
In the trade context, $i$ and $j$ may denote importer and exporter countries, and $l$ may denote industry, product category, or time in the case of repeated observations.  
Beyond the three-way structure, note that model~\eqref{eq:model_threeway} generalizes the baseline dyadic model~\eqref{eq:dyadic} by allowing for clusters of different sizes.  In the following sections, we discuss several representative scenarios.
\subsection{Random effects}
Suppose that the errors obey a full random-effects structure:
\[
 \eps_{ijl} = \eta_i + \xi_j + \zeta_l + u_{ijl}\;,
\]
where, following~\eqref{eq:random_eff1}, the random effects $\{\eta_i\}_{i\in[m]}, \{\xi_j\}_{j\in[n]}, \{\zeta_l\}_{l\in[\ell]}$ and the idiosyncratic errors $\{u_{ijl}\}_{i\in[m],j\in[n],l\in[\ell]}$ are i.i.d. within each family and mutually independent. Then, the errors satisfy three-way exchangeability: for any permutations $\pi, \sigma, \psi$ on $[m]$, $[n]$, and $[\ell]$, we have
\begin{equation}\label{eq:INVA}
    (\eps_{ijl})_{i\in[m],j\in[n],l\in[\ell]} \stackrel{d}{=} (\eps_{\pi(i)\sigma(j)\psi(l)})_{i\in[m],j\in[n],l\in[\ell]} \mid \XX,\DD\;,\tag{InvA}
\end{equation}
where $\XX, \DD$ denote the stacked matrix (or vector) for $x_{ijl}$ and $d_{ijl}$, respectively. To test $H_0\!:\beta=0$ under this invariance, we can adapt our main test as follows:
\begin{enumerate}[(i)]
    \item Apply Algorithm~\ref{alg:construct} three times to construct a three-way permutation group $\GG$.
    \item Apply Procedure~\ref{proc1} using the group $\GG$.
\end{enumerate}
The resulting permutation test is finite-sample valid.

\subsection{Panel models}
In panel models, one clustering dimension (e.g., $l$) corresponds to time:
\[
y_{ijt} = x_{ijt}^\top \gamma + d_{ijt}^\top\,\beta + \eps_{ijt}\;.
\]
In the trade model, where $i$ and $j$ denote importer and exporter countries, $t$ may represent the year of observation.  
In such settings, full three-way exchangeability~\eqref{eq:INVA} is implausible because errors are typically not exchangeable over $t$ due to autocorrelation.  
Suppose instead that the errors remain exchangeable across the first two dimensions:
\begin{equation}\label{eq:INVB}
    (\eps_{ijt})_{i\in[m],j\in[n]} \stackrel{d}{=} (\eps_{\pi(i)\sigma(j)t})_{i\in[m],j\in[n]} \mid \XX,\DD\;.\tag{InvB}
\end{equation}
This holds, for instance, under the model 
$\eps_{ijt} = \eta_i + \xi_j + \zeta_t + u_{ijt}$, 
where $\eta_i, \xi_j$ are random effects as before, but $\zeta_t$ represents an arbitrary time trend.  
To test the null hypothesis for $\beta$ under~\eqref{eq:INVB}, we proceed as follows:
\begin{enumerate}[(i)]
    \item Apply Algorithm~\ref{alg:construct} twice—once for each of the index sets $[m]$ and $[n]$—to construct a two-way permutation subgroup $\GG$.
    \item Apply Procedure~\ref{proc1} with respect to $\GG$.
\end{enumerate}
The resulting permutation test is finite-sample valid because it applies the baseline dyadic test within each time period $t$, thereby controlling for the unknown time trend.  
To our knowledge, this is the first instance of a finite-sample valid test for $\beta = 0$ under exchangeable errors \eqref{eq:INVB} in panel data models, which we discuss further in Section \ref{sec:conclusion}.

\subsection{Two-way layouts}  
In some settings, the number of observations per cell $(i,j)$ in model~\eqref{eq:model_threeway}, denoted $\ell_{ij}$, may vary across $i$ and $j$.  
In such cases, the above procedure cannot be applied directly, as it requires data blocks of equal size.

An alternative is to perform permutations only with respect to cluster $l$, effectively permuting data independently within each cell $(i,j)$.  
The full procedure is as follows:
\begin{enumerate}[(i)]
    \item For each cell $(i,j)$, apply Algorithm~\ref{alg:construct} on index set $[\ell_{ij}]$.  
    Denote the resulting permutation group as $\GG_{ij} = \{\pi_k^{ij}\}_{k=0}^K$, where $\pi_0^{ij} = \mathrm{Id}$.
    \item Apply Procedure~\ref{proc1} using $\GG$ as the permutation group, where $\GG$ is formed by concatenating $(\GG_{ij})_{i\in[m], j\in[n]}$ as in Procedure \ref{proc:CIPT}.
\end{enumerate}
The resulting permutation test is finite-sample valid, even under a relaxation of~\eqref{eq:INVA} in which exchangeability holds only with respect to cluster $l$.  
For example, the error structure $\eps_{ijl} = \eta_{ij} + \zeta_l + u_{ijl}$ satisfies this condition if $\eta_{ij}$ are arbitrary but $\zeta_l$ are i.i.d.  
This test is thus appropriate in settings where $l$ indexes independent replications, such as two-way layouts in randomized experiments~\citep{montgomery2017design}.

\subsection{Irregular designs} 
In certain two-way layouts, permuting data within each cell $(i,j)$ may not be feasible.  
For instance, if cluster $l$ denotes time, as in the panel case (B), then invariance~\eqref{eq:INVB} is not plausible, and cell-level permutations would lead to an invalid test.  
Moreover, if $d_{ijl}$ is fixed within a given cell $(i,j)$---that is, a dyad-level covariate---then within-cell permutations yield a trivial test with zero power.

A practical solution for such irregular designs is to combine the conditional permutation test for missing data (Section~\ref{sec:missing}) with the panel-data test from Case~(B).  For a fixed threshold $L_0$, define
\[
    M_{ij} = 
    \begin{cases}
        1 & \text{if } \ell_{ij} \ge L_0,\\[4pt]
        0 & \text{otherwise.}
    \end{cases}
\]
We then propose the following procedure:

\begin{enumerate}[(i)]
    \item Apply Algorithm~\ref{alg:biclique} using $M$ as the mask matrix to obtain a set of fully observed blocks, $\FOB_M$. For each cell $(i, j)$ in $\FOB_M$, randomly drop $\ell_{ij} - L_0$ observations to ensure there are exactly $L_0$ observations in each cell.
    \item Apply Procedure~\ref{proc:CIPT} to the remaining data using $\FOB_M$ from step 1 as input. 
\end{enumerate}

The schematic below visualizes the procedure through a $2\times 3$-cell toy example for $L_0=5$. The number in each cell $(i,j)$ represents the number of observations $\ell_{ij}$ in that cell. Step (i) identifies the active cells used to construct the fully observed blocks and removes unnecessary observations, as indicated by the shaded region. Step (ii) then applies Procedure~\ref{proc:CIPT} based on these fully observed blocks.
\begin{figure}[h]
  \centering
  \begin{tikzpicture}[
      every matrix/.style={
        matrix of nodes,
        nodes in empty cells,
        nodes={minimum size=8mm, draw, anchor=center},
        column sep=-\pgflinewidth,
        row sep=-\pgflinewidth
      }
    ]
    \matrix (c1) at (-5cm,0) {
      5 & 15 & 0 \\
      10 & 5 & 1 \\
    };
    \draw[->] (-3.5,0) -- (-2,0) node[midway,above,sloped]{\small Step (i)};
    \matrix (c2) at (-.5,0) {
      5 & 5 & 0 \\
      5 & 5 & 1 \\
    };
    \begin{pgfonlayer}{background}
      \foreach \i/\j in {1/1,1/2,2/1,2/2} {%
        \fill[pattern=north east lines, pattern color=black!70]
          (c2-\i-\j.north west) rectangle (c2-\i-\j.south east);
      }
    \end{pgfonlayer}
    
    \draw[->] (1,0) -- (2.5,0) node[midway,above,sloped]{\small Step (ii)};
    \matrix (c3) at (4, 0) {
      5 & 5 & 0 \\
      5 & 5 & 0 \\
    };
    \begin{pgfonlayer}{background}
      \foreach \i/\j in {1/1,1/2,2/1,2/2} {%
        \fill[pattern=north east lines, pattern color=black!70]
          (c3-\i-\j.north west) rectangle (c3-\i-\j.south east);
      }
    \end{pgfonlayer}
  \end{tikzpicture}
  \label{fig:perm_twoway}
\end{figure}

This procedure yields a finite-sample valid test by the same argument as in Theorem~\ref{thm:cond_valid}.  
However, the choice of $L_0$ introduces a power trade-off.  
Choosing a smaller $L_0$ retains more cells in $M$ but reduces within-cell sample size, leading to efficiency loss.  
Conversely, a larger $L_0$ ensures larger cell sizes but fewer eligible cells, again reducing efficiency.  
We therefore recommend tuning $L_0$ (e.g., via a grid search) to balance these trade-offs and minimize the loss of observations, thereby maximizing power as implied by Theorem~\ref{thm:power_miss}.  
A full power characterization as a function of $L_0$ remains an interesting theoretical question for future work. Additionally, the procedure involves random deletion of samples, so different runs may yield different $p$-values. As noted in Remark~\ref{rmk:random}, one may mitigate this variability by taking the median of multiple $p$-values.

Although our exposition focuses on two-way and three-way clustered errors, our procedures extend naturally to general multi-way clustering by constructing appropriate permutations along each clustering dimension. Missing data scenarios can be handled analogously using the conditional testing strategy described in Procedure \ref{proc:CIPT}. In sum, the proposed approach delivers a unified permutation-based testing framework under multi-way clustering.

\section{Simulation Studies}\label{sec:simu}
In this section, we assess the type~I and type~II errors of our permutation tests for multi-way clustering and compare their performance against existing methods.

\subsection{Type I Error Control under Dyadic Regression}
We first evaluate the type~I error in the dyadic regression model \eqref{eq:dyadic} given $\beta = 0$. Additional results for more general settings---such as irregular designs---are reported in Section~\ref{sec:trust} of the Appendix and lead to similar conclusions. To construct doubly exchangeable errors that are not trivially i.i.d., we follow prior work \citep{mackinnon2021wild} and adopt a random feature model of the form
\begin{equation}\label{eq:simu}
    \mathrm{variable}_{ij} \;=\; \sigma_1 v_{1,i} \;+\; \sigma_2 v_{2,j} \;+\; v_{3,ij}\;,
\end{equation}
where $\mathrm{variable}_{ij}$ denotes a generic variable in the dyadic regression, and $v_{1,i}$, $v_{2,j}$, and $v_{3,ij}$ are i.i.d.\ random variables drawn from a specified distribution $\mathbb{P}_v$. We induce cluster correlations $\phi_1$ and $\phi_2$ in the first and second dimensions, respectively, by setting
\[
\sigma_1^2 \;=\; \frac{\phi_1}{1 - \phi_1 - \phi_2}\;, 
\qquad
\sigma_2^2 \;=\; \frac{\phi_2}{1 - \phi_1 - \phi_2}\;.
\]
The simulation setup is as follows:
\begin{itemize}
    \item $p = 3$ and $x_{ij} = (1, z_i, z_j)^\top$ where $\{z_i\}_{i\in[n]} \stackrel{iid}{\sim} \mathrm{Unif}([0, 2])$. We set the nuisance coefficients $\gamma = (0.5, 1, 1)^\top$ and vary $n = 25, 30, \dots, 45$.
    \item Generate $w_{ij}$ from the random feature model \eqref{eq:simu} with $\P_v = \mathcal{N}(0, 1)$ and $\phi_1 = \phi_2 = 0.4$. We consider dyadic covariates following (i) a lognormal distribution $d_{ij} = \exp(0.5 w_{ij})$, and (ii) a normal distribution $d_{ij} = w_{ij}$.
    \item Generate dyadic errors $\eps_{ij}$ from Model \eqref{eq:simu} with $\P_v = \mathcal{N}(0, 1)$. Set $\phi_1 = 0.05$ and vary $\phi_2 = 0.15, 0.9$ to test different degrees of error correlation.
\end{itemize}

In this setting, we assume no missing data and implement the permutation test in Procedure~\ref{proc1}. For brevity, we refer to our testing procedure as the
“invariant permutation test” (IPT), emphasizing that it depends on the invariance assumption in~\eqref{eq:ex1}. For comparison, we implement the methods analyzed in \citet{mackinnon2021wild}, namely $t_2$, $t_3$, and WCB. The $t_2$ and $t_3$ tests are $t$-tests using different variance estimators for dyadic regression, while WCB (Wild Cluster Bootstrap) is a bootstrap procedure built on $t_3$, proposed as a more robust alternative.

Table~\ref{tab:type1} reports the type~I errors under different simulation setups, demonstrating the robustness of our permutation test. Under normal covariates and $\phi_2=0.15$, most type~I errors are close to the nominal level, indicating that all tests maintain correct size control except for $t_3$, which exhibits finite-sample size inflation, consistent with the findings of \citet{mackinnon2021wild}.
Increasing $\phi_2$ from 0.15 to 0.9 causes $t_2$ and WCB to exhibit inflated type~I errors, while our IPT procedure remains valid. Note that IPT has empirical type I errors strictly below 5\%, as the minorization step yields mildly conservative $p$-values. This conservativeness can be mitigated through the idea in Remark \ref{rmk:conservative}. For lognormal covariates in Panel B, we observe similar type~I error inflation in $t_2$, $t_3$, and WCB under different error correlations $\phi_2$.
\renewcommand{\arraystretch}{1.1}
\begin{table}[h]
    \centering
    \begin{tabular}{c| p{2em}p{2em}p{2em}p{2em}p{2em}|p{2em}p{2em}p{2em}p{2em}p{2em}}
    & \multicolumn{5}{c|}{$\phi_2 = 0.15$} 
    &\multicolumn{5}{c}{$\phi_2 = 0.9$} \\
    \hline
    & \multicolumn{5}{c|}{$n$} 
    &\multicolumn{5}{c}{$n$} \\
    \multicolumn{1}{c|}{} & 25 & 30 & 35 & 40 & 45 & 25 & 30 & 35 & 40 & 45\\
    \hline
    \multicolumn{5}{l}{Panel A: Normal Cov.} &  \multicolumn{5}{c}{}  \\
    \hline
    IPT & 1.18 & 1.57 & 1.49 & 1.68 & 1.99 & 1.16 & 1.45 & 1.56 & 1.43 & 2.06 \\ 
    $t_2$ & 3.08 & 3.04 & 3.63 & 3.64 & 3.73 & 6.16 & 6.01 & 6.04 & 5.66 & 6.01 \\ 
    $t_3$ & 7.74 & 6.88 & 7.46 & 6.56 & 6.48 & 7.84 & 7.29 & 7.23 & 6.61 & 6.76 \\ 
    WCB & 6.32 & 5.88 & 6.66 & 6.14 & 6.34 & 9.49 & 8.85 & 8.60 & 8.12 & 8.01 \\ 
    \hline
    \multicolumn{5}{l}{Panel B: Lognormal Cov.} & \multicolumn{5}{c}{}  \\
    \hline
    IPT & 1.73 & 1.70 & 1.68 & 1.83 & 2.50 & 1.16 & 1.37 & 1.52 & 1.35 & 1.90\\ 
    $t_2$ & 5.38 & 5.29 & 5.03 & 5.01 & 4.82 & 7.84 & 7.47 & 8.08 & 7.47 & 7.11\\ 
    $t_3$ & 14.72 & 13.95 & 12.95 & 12.29 & 11.57 & 11.64 & 10.83 & 11.50 & 10.36 & 9.98\\ 
    WCB & 5.12 & 4.82 & 4.81 & 4.97 & 5.09 & 7.23 & 6.76 & 7.16 & 7.22 & 6.98\\
    \end{tabular}
    \bigskip
    \caption{Type~I errors (\%) based on 10,000 simulations at the 5\% significance level.}
    \label{tab:type1}
\end{table}

\subsection{Test Power under Dyadic Regression.} 
Next, we examine the case where the null hypothesis $H_0:\beta = 0$ in model~\eqref{eq:dyadic} is false, and assess power across alternatives $\beta = 0.01, 0.02, \dots, 0.15$. The number of clusters is fixed at $n = 25$, and all other simulation details follow the previous section. To highlight performance under heavy-tailed covariates, we focus on the lognormal case.

As shown in Figure~\ref{fig:power}, the power curves of all methods are broadly similar, with larger values of $\phi_2$ corresponding to more challenging testing problems. While $t_2$ and $t_3$ exhibit slightly higher power across settings, their apparent gains are partly driven by the type~I error inflation observed in Table~\ref{tab:type1}, Panel B. In contrast, our permutation test attains comparable power while maintaining finite-sample validity. Additional simulation results can be found in Section \ref{sec:simu_supp} of the Appendix.

\begin{figure}[t!]
    \vspace{-7mm}
    \centering
    \subfloat[Lognormal Cov., $(\phi_1, \phi_2)=(0,0.1)$]{\includegraphics[width=.35\linewidth]{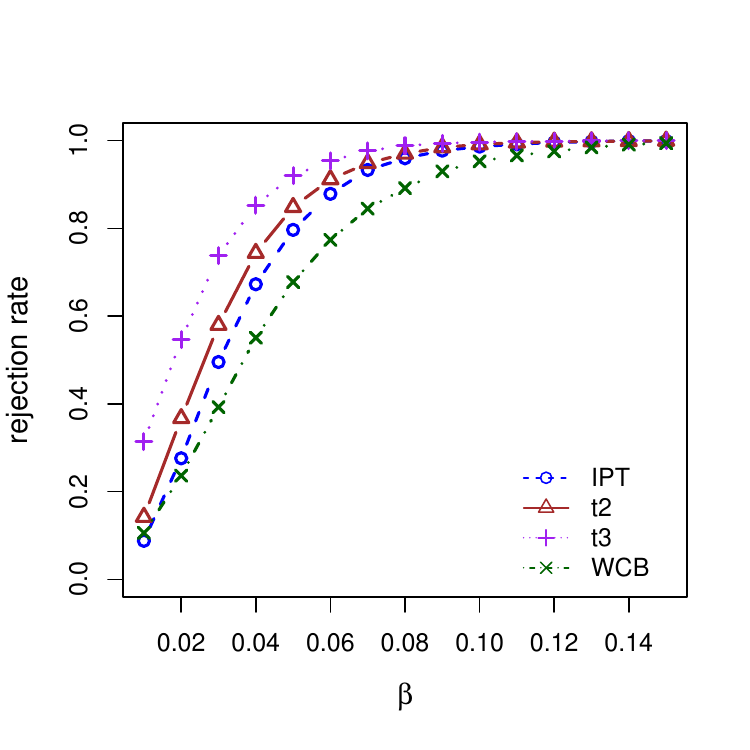}}
    \subfloat[Lognormal Cov., $(\phi_1, \phi_2)=(0,0.5)$]{\includegraphics[width=.35\linewidth]{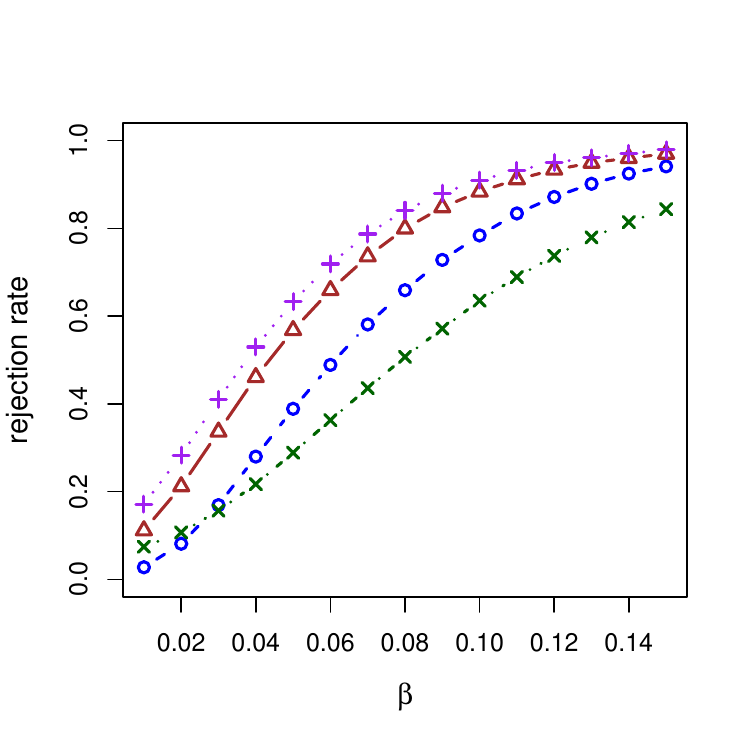}}
    \subfloat[Lognormal Cov., $(\phi_1, \phi_2)=(0,0.9)$]{\includegraphics[width=.35\linewidth]{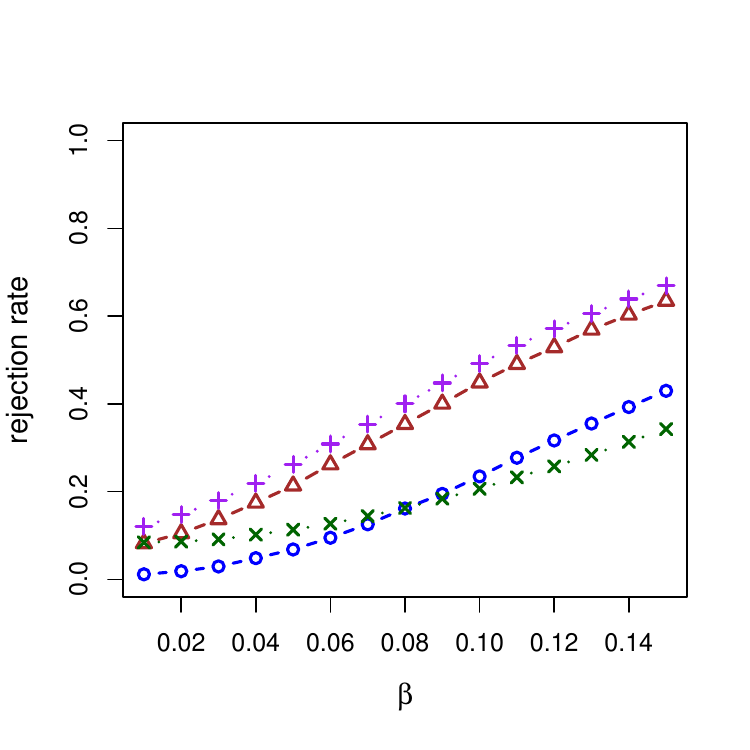}} 
    \caption{Test power based on 10,000 simulations at the 5\% significance level.}
    \label{fig:power}
\end{figure}

\section{Application: Bilateral Trade Flows}\label{sec:trade}
In a seminal paper, \citet{silva2006log} analyzed bilateral trade flows to quantify the determinants of international trade. They started with the log-linearized version of the well-known gravity model \citep{tinbergen1962shaping}:
\begin{equation}\label{eq:trade}
\log \left(1 + \mathrm{trade}_{ij}\right)
= \beta_0
+ \beta_1 \log(\mathrm{GDP}_i)
+ \beta_2 \log(\mathrm{GDP}_j)
+ \beta_3 \log(\mathrm{dist}_{ij})
+ \dots
+ \eps_{ij}\;,
\end{equation}
where $\mathrm{trade}_{ij}$ denotes exports from country $i$ to country $j$. Regressors include the logarithms of exporter and importer GDP, bilateral distance, and dyadic characteristics such as shared borders, trade openness, and trade agreements.  
\citet{silva2006log} highlighted that dyadic errors in \eqref{eq:trade} may exhibit complex dependence structures, leading to potentially severe biases in OLS estimates. To address this, \citet{silva2006log, graham2020dyadic} applied Poisson pseudo–maximum likelihood (PPML) and nonlinear least squares (NLS) estimators, which directly fit the gravity model using nonlinear methods, though they rely on asymptotic justification.

Here, we apply our IPT procedure, which is finite-sample valid under double exchangeability (Assumption~\ref{asmp:double_ex}). In the gravity model~\eqref{eq:trade}, this assumption holds, for example, if the shocks $\eps_{ij}$ are sampled i.i.d.\ across countries. We compare our results with NLS and PPML later in the section. 
We begin by constructing confidence intervals for the coefficients in \eqref{eq:trade} using four methods: $t_2$, $t_3$, WCB, and our permutation test (IPT). The regressors and confidence intervals are displayed in Figure~\ref{fig:trade}. The intervals for $t_2$ and $t_3$ rely on asymptotic normality, whereas those for WCB and IPT are obtained via test inversion (Procedure~\ref{proc1}). 

Figure~\ref{fig:trade} shows that our permutation test and WCB can produce wider confidence intervals than $t_2$ and $t_3$. This pattern aligns with our simulations, where $t_2$ and $t_3$ may achieve higher power but at the cost of inflated type~I errors. Intuitively, a smaller effective variance in a variable leads to a wider confidence interval. We illustrate this relationship using a plot of CI width versus effective variance in Section \ref{sec:trade_supp} of the Appendix.

In Table~\ref{tab:trade_variables} of Section \ref{sec:trade_supp}, we also summarize the variables identified as significant by each method. The last two columns report results from the NLS and PPML estimates in \citet{silva2006log}. Compared with $t$-tests, our permutation test identifies more variables as non-significant --- such as importer’s per-capita GDP and the openness dummy --- consistent with NLS and PPML. Our test also classifies exporter’s per-capita GDP as non-significant, whereas NLS and PPML find it significant. Further comparisons and discussion can be found in Section~\ref{sec:trade_supp}.

\begin{figure}[h]
    \centering
    \includegraphics[width=0.7\linewidth]{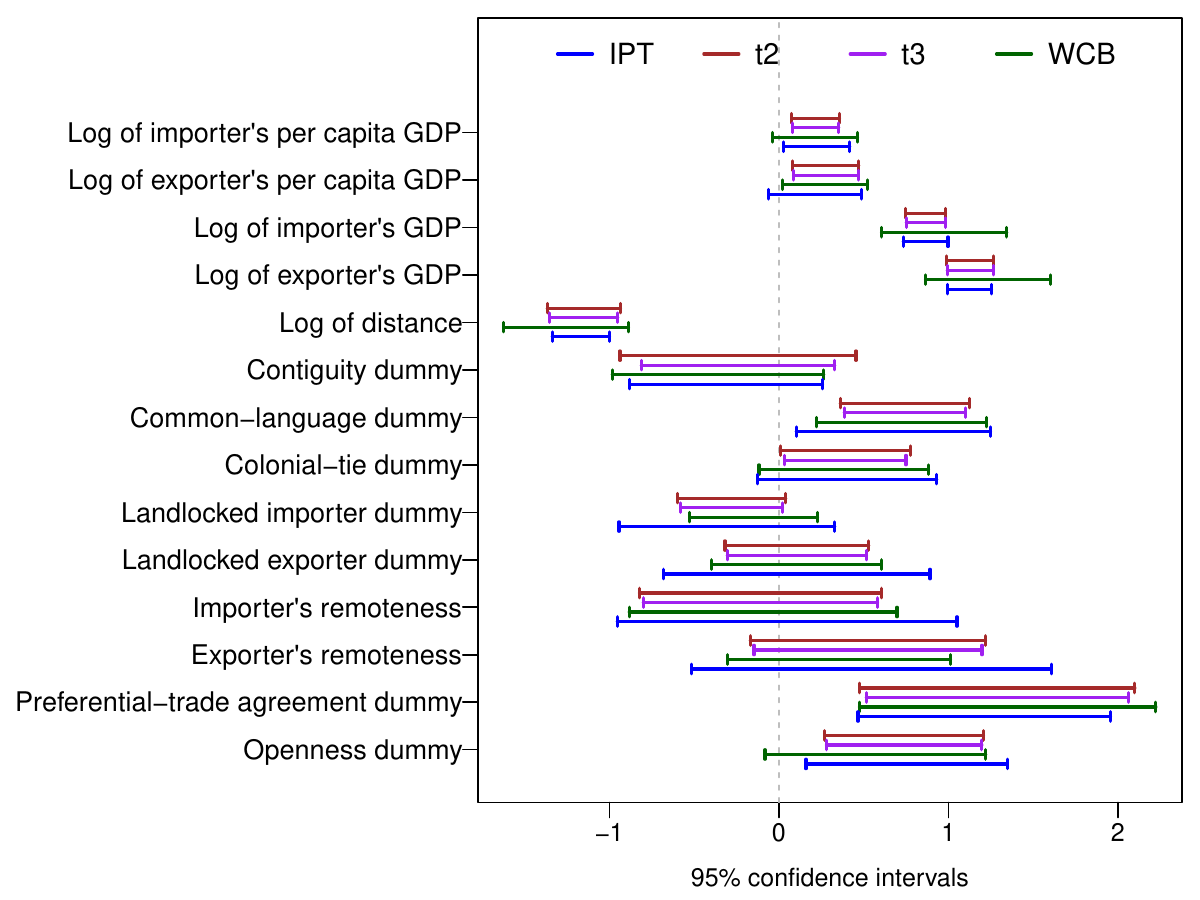}
    \caption{Confidence intervals for different variables in the regression model. We refer readers to \citet{silva2006log} for detailed variable descriptions.}
    \label{fig:trade}
\end{figure}

\section{Concluding Remarks}\label{sec:conclusion}
In this paper, we propose permutation tests of significance that are finite-sample valid under a suitable double invariance assumption on the errors. This framework of invariance-based inference provides an attractive alternative to standard asymptotic methods, particularly in settings with complex dependence structures. As we have shown, however, such complexity inevitably entails a trade-off between finite-sample validity and statistical efficiency. For example, our procedure for irregular designs—see scenario (D) in Section~\ref{sec:multicluster}—relies on subsampling the original data, which can result in efficiency loss. Similarly, our conditional permutation tests for missing data require a pre-specified mask matrix, whose construction involves balancing the number of cells against their sizes. Developing practical algorithms to optimally design this mask matrix is an important next step. More broadly, providing concrete guidance on how to navigate the trade-off between validity and efficiency is a valuable direction for future research.

An important extension is to analyze the performance of our test under the specific setting of panel data. On one hand, applying the invariance-based procedure to real panel data would provide practical insights beyond the simulations and applications considered in the current paper. On the other hand, a formal characterization of the test’s power under various panel-specific error structures remains open. It remains a compelling direction for future work to develop such theoretical guarantees and to understand how temporal dependence in panel data shapes the test’s efficiency.

\bibliographystyle{apalike}
\bibliography{ref}

\newpage
\appendix

\section{Dyadic Regression with Missing Data}\label{sec:missing}
In this section, we turn our attention to the problem of missing data. We remain focused on testing $H_0$ under the dyadic regression model \eqref{eq:dyadic}, but now only a subset of cells are observed:
\begin{equation}\label{eq:model_missing}
    (y_{i j}, x_{ij}, d_{ij}) ~\text{is observed if and only if}~ (i,j)~\text{satisfies that}~M_{ij} = 1~.
\end{equation}

In the expression above, $M \in \{0, 1\}^{n\times n}$ is a mask matrix, where $M_{ij} = 1$ indicates an observed cell $(i, j)$ and $M_{ij} = 0$ indicates a missing cell. In general, it is not possible to conduct valid inference if $M$ has an arbitrary distribution with respect to the observed data~\citep{little2019statistical}. To make progress, we make the following assumption on the missingness mechanism.
\begin{assumption}\label{A:missing}
The missing mask matrix $M$ satisfies 
$$
M \indep (\eps_{ij})_{i,j\in[n]} \mid\XX,\DD\;.
$$
\end{assumption}
In plain language, under Assumption~\ref{A:missing}, the missing mechanism is conditionally independent of dyadic errors, such that the double exchangeability is preserved conditional on the observed data. 
Note that this assumption allows the missing data mechanism to depend on $\XX,\DD$ in an arbitrary way. However, the assumption will be violated, in general, if 
$M$ depends on the observed outcomes $\yy$.

To test $H_0:\beta = 0$ under missing data, we first study why Procedure \ref{proc1} can fail. By way of illustration, in Figure \ref{fig:missing_ex} we visualize how row and column permutations in Procedure \ref{proc1} work under the missing mask. In the toy example with four clusters depicted in the figure, we observe that some permutations (e.g., permutation 1) become infeasible, as we permute observed entries with missing entries. However, other permutations (e.g., permutation 2) remain feasible because they only permute data between fully observed entries.

\begin{figure}[h]
  \centering
  \begin{tikzpicture}[
      every matrix/.style={
        matrix of nodes,
        nodes in empty cells,
        nodes={minimum size=9mm, draw, anchor=center},
        column sep=-\pgflinewidth,
        row sep=-\pgflinewidth
      },
      >=stealth,
      scale=0.95
    ]

    \matrix (M) at (-6.0,0) {
      1 & 0 & 0 & 0\\
      0 & 1 & 1 & 1\\
      0 & 1 & 1 & 1\\
      0 & 1 & 1 & 1\\
    };
    \node[above=2mm of M] {\bfseries Missing Mask};

    \matrix (Y) at (-2.0,0) {
      $y_{11}$ & $\mathrm{NA}$ & $\mathrm{NA}$ & $\mathrm{NA}$ \\
      $\mathrm{NA}$ & $y_{22}$ & $y_{23}$ & $y_{24}$\\
      $\mathrm{NA}$ & $y_{32}$ & $y_{33}$ & $y_{34}$\\
      $\mathrm{NA}$ & $y_{42}$ & $y_{43}$ & $y_{44}$ \\
    };
    \node[above=2mm of Y] {\bfseries Outcome Matrix};

    \draw[->] (0.8,0.5) -- (2.6,2.8) node[midway,above,sloped]{\small permutation 1};

    \draw[->] (0.8,-0.5) -- (2.6,-2.8) node[midway,below,sloped]{\small permutation 2};

    \matrix (G) at (6.0,2.5) {
      $y_{33}$ & $y_{32}$ & $\mathrm{NA}$ & $y_{34}$ \\
      $y_{23}$ & $y_{22}$ & $\mathrm{NA}$ & $y_{24}$\\
      $\mathrm{NA}$ & $\mathrm{NA}$ & $y_{11}$ & $\mathrm{NA}$\\
      $y_{43}$ & $y_{42}$ & $\mathrm{NA}$ & $y_{44}$ \\
    };
    \node[above=2mm of G] {\bfseries General Permutation (invalid)};
    \begin{pgfonlayer}{background}
      \foreach \i/\j in {2/1,1/2,3/2,2/3,1/4,4/1,3/4,4/3} {%
        \fill[pattern=north east lines, pattern color=black!70]
          (G-\i-\j.north west) rectangle (G-\i-\j.south east);
      }
    \end{pgfonlayer}
    \matrix (R) at (6.0,-2.5) {
      $y_{11}$ & $\mathrm{NA}$ & $\mathrm{NA}$ & $\mathrm{NA}$ \\
      $\mathrm{NA}$ & $y_{33}$ & $y_{32}$ & $y_{34}$\\
      $\mathrm{NA}$ & $y_{23}$ & $y_{22}$ & $y_{24}$\\
      $\mathrm{NA}$ & $y_{43}$ & $y_{42}$ & $y_{44}$ \\
    };
    \node[above=2mm of R] {\bfseries Restricted Permutation (valid)};
  \end{tikzpicture}
  \vspace{-1em}
  \caption{{\em Left:} The missing mask matrix $M$ induces outcomes with missing values~(marked ``$\mathrm{NA}$"). {\em Right:} Permutation 1 swaps columns 1 and 3 and rows 1 and 3, thus permuting observed values with missing values as indicated by the shaded cells. In contrast, permutation 2 swaps columns 2 and 3 and rows 2 and 3. This only permutes data from fully observed entries. Procedure~\ref{proc1} is infeasible under permutation 1, but feasible under permutation 2.}
  \label{fig:missing_ex}
\end{figure}

The following definition aims to formalize the intuition conveyed in Figure~\ref{fig:missing_ex}.
\begin{definition}\label{def:perm_block}
Let $M$ be a given mask matrix. For $I, J \subseteq [n]$, the pair of index subsets $I\times J$ is a fully observed block under $M$ if $M_{ij} = 1$~for all~$(i,j)\in I\times J$. A set of fully observed blocks $\FOB_M = \{I_q \times J_q\}_{q=1}^Q$ is disjoint if
\[
  I_{q}\,\cap\,I_{q'} = J_{q}\,\cap\,J_{q'}
  \;=\;\emptyset
  \quad\text{for all}~q\neq q'\;.
\]
\end{definition}

Given a disjoint set of fully observed blocks $\FOB_M = \{I_q \times J_q\}_{q=1}^Q$, we may decompose the original data into sub-arrays---each determined by 
rows in $I_q$ and columns in $J_q$---and then run Procedure~\ref{proc1} separately within each sub-array. This is possible because there are no missing data within each sub-array, and errors remain doubly exchangeable under Assumption~\ref{A:missing}. 
This idea is implemented in Procedure~\ref{proc:CIPT} below.

\begin{testprocedure}[Conditional Permutation Test with Missing Data]\label{proc:CIPT}
\begin{enumerate}
    \setlength{\itemsep}{-0.5em}
    \item Construct a set of fully observed blocks, $\FOB_M = \{I_q\times J_q\}_{q=1}^Q$, under mask $M$. 

    \item For each $q=1, \ldots, Q$:
    \begin{enumerate}
        \item Subset the data by rows $I_q$ and columns $J_q$. Denote the data sub-array as $(\XX_q, \DD_q, \yy_q)$.

        \item Apply Algorithm~\ref{alg:construct} on index sets $I_q, J_q$ to generate the permutation set $\GG_q = \{(\pi_{0,q}, \sigma_{0,q}), \ldots, (\pi_{K,q}, \sigma_{K,q})\}$ for data sub-array $q$.

    \end{enumerate}

    \item Construct $\GG =\{(\pi_0, \sigma_0), \ldots, (\pi_{K}, \sigma_{K})\}$ 
    by concatenating the sub-array permutations,
    $\pi_k = (\pi_{k,1}:\ldots:\pi_{k, Q})$,  for $k=0, \ldots, K$.\footnote{In this notation,  $\pi_k$ is a permutation of $[n]$, such that elements in $I_q$ are permuted according to $\pi_{k,q}$, for each $q=1, \ldots, Q$, while all other elements of $[n]$ remain fixed.} 
    Similarly, define $\sigma_k = (\sigma_{k,1}:\ldots:\sigma_{k, Q})$.
    \item Apply Procedure~\ref{proc1} to the stacked data $\{(\XX_q, \DD_q, \yy_q)\}_{q=1}^Q$ and the
    permutation set $\GG$ from step 3.
\end{enumerate}
\end{testprocedure}

 The following theorem establishes the finite-sample validity of Procedure~\ref{proc:CIPT}. 
 See Section \ref{sec:proof_CIPT} for a proof.
\begin{theorem}\label{thm:cond_valid}
Suppose that Assumptions \ref{asmp:double_ex} and~\ref{A:missing} hold for the dyadic regression model~\eqref{eq:dyadic}. Suppose also that $p<\widetilde N/2$, where 
$\widetilde N = \sum_{q=1}^Q |I_q|\cdot |J_q| := |\FOB_M|$. Then, under $H_0: \beta = 0$, the $p$-value obtained from Procedure \ref{proc:CIPT} satisfies:
\begin{equation*}
    \mathbb{P}(\mathrm{pval} \leq \alpha \mid M, \XX,\DD) \le \alpha~\text{for all}~\alpha \in[0,1],~n > 0\;.
\end{equation*}
\end{theorem}
To maintain finite-sample validity, Procedure~\ref{proc:CIPT} uses a subset of the data contained in the fully observed blocks $\FOB_M$, but this may result in efficiency loss. We will study such power loss in Section~\ref{sec:power_CRP} under  assumptions on the missingness mechanism.

\subsection{A Graph-based Construction of Fully Observed Blocks}\label{sec:fob}
Here, we explain the graph-theoretic construction of $\FOB_M$. For any given mask matrix, $M$, we may view $M$ as an adjacency matrix that defines edges between rows and columns; i.e., row $i$ is connected with column $j$ if and only if $M_{ij}=1$. In this representation, the node set of the graph consists of row and column nodes, defined by $U = [n]$ and $V = [n]$.
Then, $G = (U, V, M)$ is a graph representation of the missing data problem where $M$ characterizes the edge set. Note that $G$ is a bipartite graph as all edges run strictly between $U$ and $V$.

In this context, consider a set of fully observed blocks, $\FOB_M = (I_q \times J_q )_{q=1}^Q$. Then, for any $q=1, \ldots, Q$, the nodes $(I_q, J_q)$ form a {\em clique} in graph $G$---more precisely, it is a biclique since we define 
edges only between a row and a column. This biclique property follows directly from the definition of $\FOB_M$, since the data are observed ($M_{ij}=1$ for $i\in I_q, j\in J_q$) in all row-column pairs defined by $I_q, J_q$.

As a consequence, the problem of constructing $\FOB_M$ is equivalent to finding a {\em biclique decomposition} of $G$.
Motivated by this, we propose Algorithm \ref{alg:biclique} that sequentially searches for large bicliques on the design graph $G$, in an effort to maximize power.
\begin{algorithm}
\DontPrintSemicolon
\KwData{Design graph $G = (U, V, M)$.}
\textbf{Begin:} \\
$\FOB_M \gets \{\}$\;
\While{$M$ has nonzero entries}{
Solve the ``largest biclique problem'':
\begin{equation}\label{eq:largeclique}
    (I_*, J_*) = \underset{(I, J) ~\text{forms a biclique of $G$}}{\arg\max} |I| |J|\;.
    \vspace{-1.5em}
\end{equation}\;
Remove all edges connected to this biclique: $M_{ij} \gets 0$ if $i \in I_*$ or $j \in J_*$\;
Update the biclique set: $\FOB_M \gets \FOB_M \cup \{I_* \times J_*\}$\;
}
\KwResult{A biclique decomposition $\FOB_M$.}
\caption{Biclique search algorithm to construct $\FOB_M$.}\label{alg:biclique}
\end{algorithm}

The main computational challenge of Algorithm \ref{alg:biclique} is in Equation \eqref{eq:largeclique}, where we calculate bicliques with the largest possible number of edges in the remaining design graph. This is known as the {\em maximum edge biclique problem}, and is generally NP-hard~\citep{peeters2003maximum}; see also \citet{zhang2014finding} for a review. Fortunately, our testing procedure remains valid even when the solution to \eqref{eq:largeclique} is approximate, so it is not necessary to solve the optimization problem exactly. In the above algorithm, we propose the {\tt Bimax} method \citep{prelic2006systematic} for such approximation. 

\begin{remark}
The biclique decomposition problem also arises in~\citet{puelz2021graph} in the context of randomization tests under interference. In that paper, a biclique represents a subset of the data for which a weak null hypothesis reduces to a sharp one that can be tested via a Fisher randomization test. This interpretation is conceptually distinct from our setting, where a biclique corresponds to a fully observed subset of the data.
Nevertheless, in both contexts, achieving finite-sample validity requires conditioning on a biclique, so the task of identifying bicliques is common to both approaches.
\end{remark}

\subsection{Power Analysis of Procedure~\ref{proc:CIPT}}\label{sec:power_CRP}
Similar to Section \ref{sec:power}, we focus on the $d = 1$ setting of testing a single coefficient. The power properties of Procedure~\ref{proc:CIPT} necessarily depend on the missing-data mechanism. Because this mechanism can be arbitrary, we impose additional assumptions to enable theoretical analysis. Specifically, we assume that the missing cells are generated independently with missingness probability \(1 - \rho\), corresponding to the {\em missing completely at random} (MCAR) scenario of~\citet{little2019statistical}.  
We therefore proceed with the following assumption.

\begin{assumption}\label{A:random-graph}
The design graph $G$ is generated from mask $M$, where $\{M_{ij}\}_{i,j\in[n]}$ are independent Bernoulli variables with probability $\P(M_{ij}=1) = \rho$.
\end{assumption}
The type of random bipartite graphs related to Assumption~\ref{A:random-graph} has been widely studied in random graph theory \citep{frieze2015introduction}, and can be viewed as a natural generalization of Erdős–R\'{e}nyi-type graphs~\citep{erdos1960evolution}. Under this random graph model, we establish the following result on the size of the largest biclique, which is a determinant factor for the test power.
\begin{proposition}\label{thm:biclique}
Under Assumption \ref{A:random-graph}, with probability converging to one, the largest biclique in $G$ has size $s\times s$ with 
\begin{equation*}
    s = \Omega\Bigl( \min\Bigl\{\frac{\rho}{1 - \rho}\log n, \frac{n^{1/3}}{\log n}\Bigr\}\Bigr)\;.
\end{equation*}
\end{proposition}
The rate of the largest biclique depends on $\rho / (1 - \rho)$, the probability odds of observing any data cell---let $r_{\mathrm{odds}} \coloneqq \rho/ (1-\rho)$. Proposition \ref{thm:biclique} is motivated by the classical literature on the maximum clique size in random graphs, and its proof relies on concentration bounds on clique counts. More concretely, the first component ($\log n$) is consistent with the classical result that the maximum clique in an \ER~graph grows at the order of $\log n$ \citep{matula1976largest}; additionally, we quantify how the clique size grows with $\rodds$. The second component, to our best knowledge, is a new result characterizing the limiting behavior of the maximum clique, with the rate $n^{1/3}$ stemming from the technical proof.

Next, we apply Proposition~\ref{thm:biclique} to characterize the statistical power of the conditional permutation test of Procedure~\ref{proc:CIPT}.
\begin{theorem}\label{thm:power_miss}
Suppose that $K$ is fixed and that Assumptions \ref{asmp:double_ex}-\ref{asmp:err}, \ref{A:random-graph} hold. Let $\mathrm{pval}$ be the randomization $p$-value from Procedure \ref{proc:CIPT} on a single permutation block from the largest biclique of the design graph. 
In the asymptotic regime where $\beta$ and $p$ vary with $n$ in a way such that $s>(\max\{2/\kappa, 2\}+C)p$ for some constant $C>0$, $s$ is divisible by $K+1$, and
$$
|\beta|=\Omega\left(s^{-\frac{t}{1+t}}\right)~\text {if}~t<1 \quad~\text {or}~ \quad|\beta|=\omega\left(s^{-\frac{1}{2}}\right)~\text{if}~t=1\;,
$$
we have $\lim_{n \to \infty} \P(\mathrm{pval} > \frac{1}{K+1} \mid \XX,\DD )=0$.
\end{theorem}

Theorem~\ref{thm:power_miss} follows directly from Theorem~\ref{thm:power} by replacing the sample size \(n\) with the largest biclique size \(s\) defined in Proposition~\ref{thm:biclique}.
Although Procedure~\ref{proc:CIPT} allows a sequence of disjoint bicliques in practice, we condition on a single largest clique to simplify the analysis. 
Under missing data, the minimum detectable signal strength scales with the missingness factor captured by \(s\).
To better understand these results, consider the following scenarios:

\begin{enumerate}[(i)]
    \item {\bf Constant missingness probability.}  
    When \(\rho \in (0,1)\) is fixed, the largest biclique has size of order \(\log n \). 
    This is a well-established result in random graph theory~\citep{frieze2015introduction}.  
    Consequently, our test requires exponentially larger sample sizes to achieve a vanishing type~II error.  
    That is, loosely speaking, to detect a signal of magnitude \(|\beta| = \epsilon\), we require \(e^{1/\epsilon}\) data points. 
    This illustrates the fundamental impact of MCAR data on efficiency loss.  

    \item {\bf Diminishing missingness probability I.}  
    Suppose that \(1-\rho \to 0\) and, in addition, \(r_{\mathrm{odds}} = O(n^{1/3}/\log^2 n)\); that is, the odds of observation increase slowly at this rate.  
    Then, the biclique size satisfies
    \[
    s \sim r_{\mathrm{odds}}\log n\;.
    \]
    In other words, \(s\) grows proportionally to the odds of observation up to a logarithmic factor.  
    Compared with the first case, \(s\) can grow polynomially in \(n\),  
    and thus the minimum detectable signal can decay at a polynomial rate, though not necessarily at the canonical \(n^{-1/2}\) rate.

    \item {\bf Diminishing missingness probability II.}  
    Suppose that \(1 - \rho \to 0\) and \(r_{\mathrm{odds}} > n^{1/3}/\log^2 n\),  
    corresponding to rapidly increasing odds of observation.  
    Then, the biclique size saturates at order \(n^{1/3}/\log n\).  
    Beyond this threshold, further increases in the odds of observation do not enlarge the biclique size.
\end{enumerate}

We note that our power analysis under missing data is novel and potentially of independent interest.  
Prior work on random bipartite graphs has analyzed biclique sizes primarily under the assumption of a constant edge probability~\(\rho\)~\citep{frieze2015introduction, matula1976largest}.  
Our results extend this setting in two directions.  
First, we establish a conceptual link between missing-data mechanisms and random-graph theory, which proves useful for conducting power analysis under missingness.  
Second, we refine the graph-theoretic analysis by allowing \(\rho\) to vary with \(n\), yielding a more complete characterization of how missingness affects biclique size—namely, the second and third regimes discussed above.  
Additional details and technical proofs are provided in Section~\ref{sec:proof_CIPT_power}.

\section{Extensions of IPT}\label{sec:extension}
\subsection{Invariance-based Testing under Error Heteroskedasticity}\label{sec:sign}
Under the double sign symmetry in Assumption~\ref{asmp:double_sign}, we can construct a finite group of sign-flipping transformations. With a slight abuse of notation, let
$\{(\pi_k,\sigma_k)\}_{k=0}^K$ denote a collection of paired sign-flipping operations on $[n]$, where each $\pi_k,\sigma_k\in\{-1,1\}^n$, and suppose that the collection forms a group under componentwise multiplication. Given a pair of sign-flipping operations $(\pi,\sigma)$, let $\yy_{\pi,\sigma}$ denote the transformed outcomes obtained by replacing each $y_{ij}$ with $\pi_i \sigma_j y_{ij}$. We define $\DD_{\pi,\sigma}$ and $\XX_{\pi,\sigma}$ analogously. We can then apply Procedure~\ref{proc1} without further modification, replacing the permutation transformations with the corresponding sign flips, to obtain a finite-sample valid test.

Concretely, each sign-flipping group can be constructed as in Algorithm \ref{alg:sign}, adapting the sign-flipping scheme in wild cluster bootstrap.

\begin{algorithm}
\DontPrintSemicolon
\KwData{An index set $I$ and the number of sign flips $K$, where $K+1=2^c$ for some integer $c$.}
\Begin{
Randomly partition $I$ into $c$ disjoint subsets and let $g: I \to[c]$ denote the partition mapping\;
Construct the set $\cB=\{-1,1\}^c$ of all sign vectors of length $c$. Enumerate its $2^c=K+1$ elements as $s_0,s_1,\ldots,s_K$, with $s_0=(1,\ldots,1)$\;
\For{$k=0,1,\ldots,K$}{
    Define the sign-flipping operation $\psi_k \in \{-1,1\}^{|I|}$ by
    \[
    \psi_{k, i}=s_{k,g(i)}, \qquad i\in I.
    \]
    }
}
\KwResult{A set of sign-flipping operations $\{\psi_k\}_{k=0}^K$, where $\psi_0$ is the all-one vector.}
\caption{Sign-flipping set construction\label{alg:sign}}
\end{algorithm}

Intuitively, Algorithm~\ref{alg:sign} first partitions the index set into $c$ subsets and then applies sign flips at the subset level. This construction ensures that the resulting collection of sign-flipping operations forms a group, so that the validity argument in Theorem~\ref{thm:valid} can be directly extended to the sign-symmetric and heteroskedastic setting. The construction requires $K+1$ to be a power of two and we consider $K = 2^5 - 1 = 31$ as a default choice. We illustrate the performance in Section \ref{sec:simu_supp}.

\subsection{Addressing Conservativeness}\label{sec:conservative}
Procedure 1 can be conservative, especially when the number of transformations is small or when the ``minimum'' step makes the $p$-value grid coarse. As discussed in Remark~\ref{rmk:conservative}, this conservativeness is not an inherent limitation of our test. Concretely, we can replace Step 2 by a pairwise comparison following \citet{guan2024conformal}.

For $k$-th permutation in $\cG$, we consider $\sigma_k = \pi_k$ (same permutation to rows and columns) and keep the original test statistic
\[
a_k = |\langle V_k^\top \mathbf{D}, V_k^\top \mathbf{y}\rangle|\;,
\]
where $V_k$ is an orthonormal matrix satisfying $V_k^\top \XX=0$ and $V_k^\top \XX_{\pi_{k},\pi_{k}}=0$. The pairwise comparison only requires us to compute another test statistic 
\[
a_k^{\prime} = |\langle V_k^\top \mathbf{D}_{\pi_k, \pi_k}, V_k^\top \mathbf{y}\rangle|\;.
\]
Then, following \citet{guan2024conformal}, the randomization $p$-value can be constructed from 
\[
    \omega_k
    =
    \mathbf 1\{a_k^\prime > a_k \}
    +\frac12\mathbf 1\{a_k^\prime=a_k\},
    \qquad
    p_{\mathrm{pair}}
    =
    \frac{1+\sum_{k=1}^K\omega_k}{K+1}.
\]
By \citet{guan2024conformal}, under an assumption that the numbers of row and column clusters are equal, we have $\P(p_{pair} \le \alpha) \le 2\alpha$, and therefore the pairwise test is theoretically finite-sample valid at level \(\alpha\) using the test threshold $p_{\mathrm{pair}} \le \alpha/2$. 
Moreover, using the unadjusted threshold $p_{\mathrm{pair}} \le \alpha$, the test achieves good empirical type I error control. This modification is minor: it changes only the calibration step, leaving the model and the residualization idea intact. We validate this pairwise IPT test with two thresholds through simulations in Section \ref{sec:simu_supp}.

\section{Application: Trust Level Study}\label{sec:trust}

\citet{nunn2011slave} examined trust levels among different ethnic groups in Africa and found that intergroup mistrust is significantly associated with historical slave exports. Using the same data, \citet{mackinnon2021wild} applied bootstrap methods to account for clustering in the errors and, in contrast to the original findings, obtained a non-significant effect of slave exports on trust levels. Here, we apply our invariant permutation test (IPT) to provide finite-sample valid results for the trust level study.

Specifically, we focus on the following regression model from \citet{nunn2011slave, mackinnon2021wild}:
\begin{equation}\label{eq:trust}
    \text{trust}_{ijl} = \alpha_{j} + \mathrm{exports}_i \beta + x_{ijl}^{\top} \gamma + \eps_{ijl}.
\end{equation}
Here, $i$ and $j$ correspond to the $i$th ethnic cluster and the $j$th country, and $l$ indexes the $l$th observation within the $(i, j)$ cell. The variable $\mathrm{trust}_{ijl}$ is an integer-valued measure of the trust an individual has toward their neighbors. The right-hand side includes a country fixed effect $\alpha_{j}$, the volume of historical slave exports, and additional covariates $x_{ijl}$ such as age, gender, and education level. 

The dataset includes $i = 1, \dots, 186$ ethnic groups and $j = 1, \dots, 16$ countries. The number of observations $\ell_{ij}$ across cells $(i,j)$ is highly unbalanced: 92.4\% of cells are empty, the average cell size is 6.76, and the maximum size is 852. Hence, the model and data in this example fall under the irregular design (D) category discussed in Section~\ref{sec:multicluster}.

To test the null hypothesis $H_0: \beta = 0$, we apply the modified permutation procedure proposed for irregular designs of this type in Section~\ref{sec:multicluster}. In the implementation, we set $L_0 = 10$ and repeat the full permutation test 100 times to reduce randomness from subsampling, an idea explained in Remark \ref{rmk:random}. We report the median $p$-value as the final result. Table~\ref{tab:trust} reports the finite-sample $p$-value and the confidence interval from test inversion. We find a non-significant $p$-value for the coefficient of $\mathrm{exports}$ on $\mathrm{trust}$ in model~\eqref{eq:trust}, and this conclusion is robust to different choices of $L_0$. This finding is consistent with the results of \citet{mackinnon2021wild}, who also report relatively weak evidence against the null hypothesis $\beta = 0$. The key difference is that their inference is only asymptotically valid, whereas our $p$-value is valid in finite samples. We provide further robustness checks on the type I error control under irregular designs in Section \ref{sec:trust}.
\renewcommand{\arraystretch}{1.2}
\begin{table}[htbp!]
    \centering
    \begin{tabular}{c|c}
    \hline
    $p$-value for $H_0$ & Confidence interval \\
    \hline
    0.76 & [-1.19, 0.19] \\
    \hline
    \end{tabular}
    \bigskip
    \caption{Finite-sample valid $p$-value and confidence interval at the 95\% level.}
    \label{tab:trust}
\end{table}

Next, we follow the real missing pattern above and further evaluate the type I error of different methods. Specifically, we retain the observed irregular design and generate new covariates and errors as follows:
\begin{itemize}
    \item Keep the $\mathrm{exports}_i$ from the observed data and set $\beta = 0$ as we consider the test performance under the null.
    \item Generate nuisance covariates $x_{ijl}\in\R^{10}$ from the random-effects model~\eqref{eq:simu}, independently across coordinates, with correlation $\phi_1 = \phi_2 = 0.4$ along indices $i$ and $j$, respectively. Generate the corresponding $\gamma$ by $\gamma_j \stackrel{iid}{\sim} \mathrm{Binom}(3,0.3)+\cN(0,0.5^2),$ for $j = 1, \dots, 10$.
    \item We consider different error types: (I) $\eps_{ijl} = \exp(v_i u_{ijl})$, where $v_i\stackrel{iid}{\sim} \cN(0, 1)$ and $u_{ijl}\stackrel{iid}{\sim}\cN(0, 1)$, (II) $\eps_{ijl}$ are generated from the random effect model \eqref{eq:simu} with $\phi_1 = \phi_2 =0.1$ on index $i$ and index $j$, (III) $\eps_{ijl}$ are similarly generated as (II) with $\phi_1 = 0.9, \phi_2 = 0$.
\end{itemize}
These semi-synthetic data sets allow us to assess robustness of type I error control across distinct forms of error dependence. As shown in Table~\ref{tab:type1_NW}, IPT maintains correct type~I error in all scenarios, whereas all the alternative methods suffer size distortions under certain error distributions. Notably, although WCB is considered robust to data distributions \citep{mackinnon2021wild}, it is not finite-sample valid and can exhibit substantially inflated type I errors in these designs.
\renewcommand{\arraystretch}{1.1}
\begin{table}[h]
    \centering
    \begin{tabular}{l| p{2em}p{2em}p{2em}p{2em}p{2em}}
    \hline
    Error Type & (I) & (II) & (III) \\
    \hline
    IPT & 1.2 & 0.9 & 1.0\\ 
    $t_2$ & 2.3 & 5.8 & 3.4\\ 
    $t_3$ & 10.0 & 10.0 & 14.8\\ 
    WCB & 8.5 & 8.5 & 6.9\\ 
    \hline
    \end{tabular}
    \bigskip
    \caption{Type~I errors (\%) based on 1,000 simulations at the 5\% significance level. The tests are evaluated on semi-synthetic data from the trust level study.}
    \label{tab:type1_NW}
\end{table}

\section{Additional Simulation Results in Section \ref{sec:simu}}\label{sec:simu_supp}
Figure \ref{fig:type1} below visualizes the type I errors presented in Table \ref{tab:type1}. We refer readers to Section \ref{sec:trust} for type I errors under more general settings, i.e., irregular designs.
\begin{figure}[h]
    \vspace{-7mm}
    \centering
    \subfloat[Normal Cov., $\phi_2=0.15$]{\includegraphics[width=.4\linewidth]{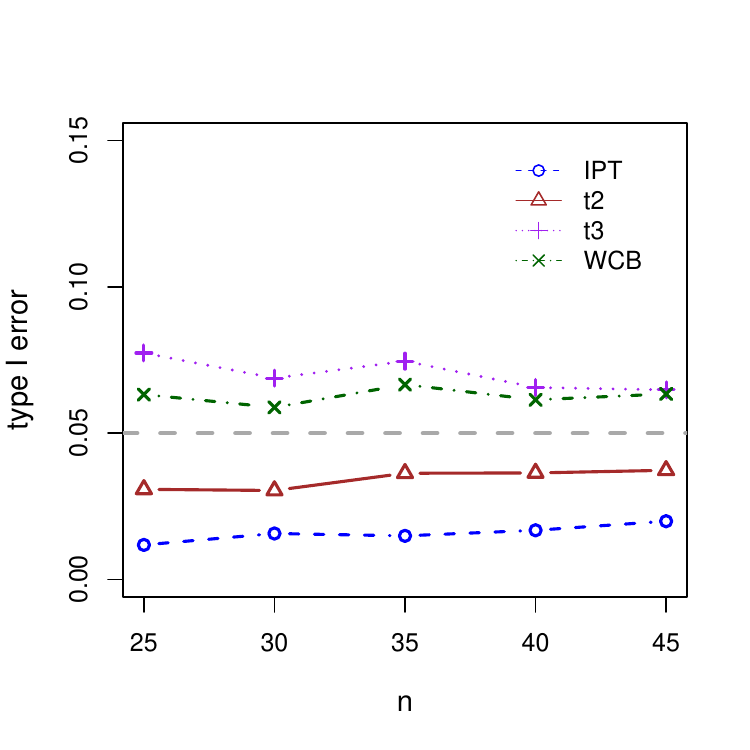}}
    \subfloat[Normal Cov., $\phi_2=0.9$]{\includegraphics[width=.4\linewidth]{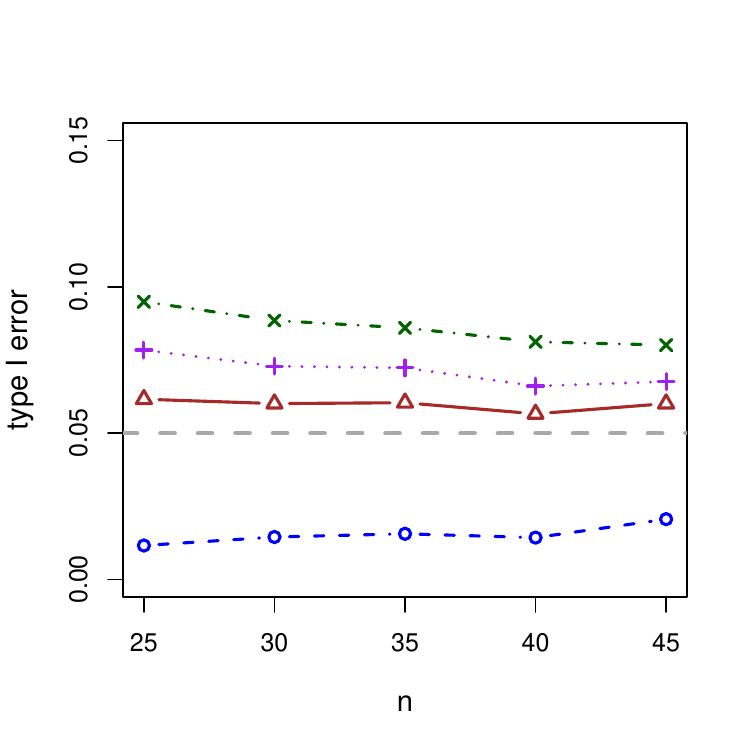}} \\
    \vspace{-2em}
    \subfloat[Lognormal Cov., $\phi_2=0.15$]{\includegraphics[width=.4\linewidth]{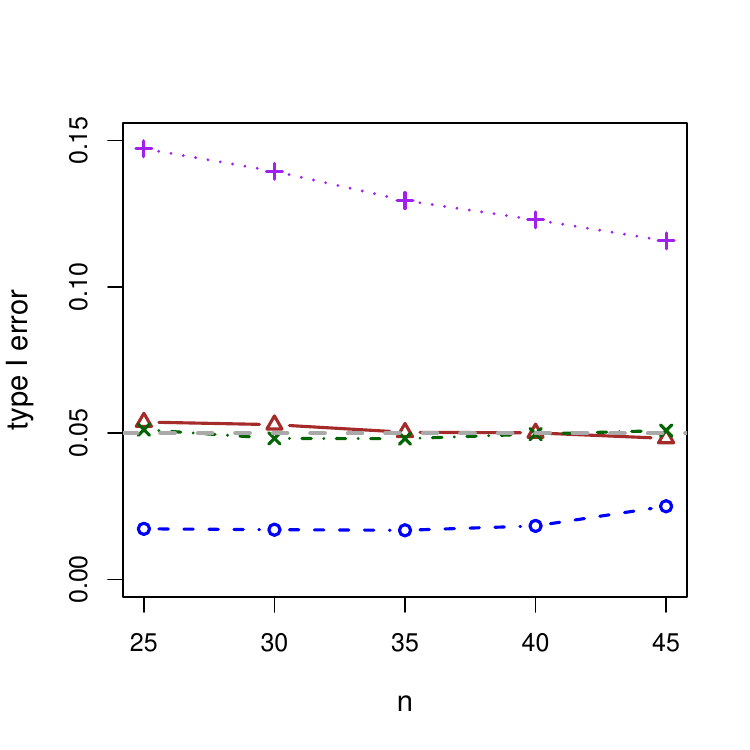}}
    \subfloat[Lognormal Cov., $\phi_2=0.9$]{\includegraphics[width=.4\linewidth]{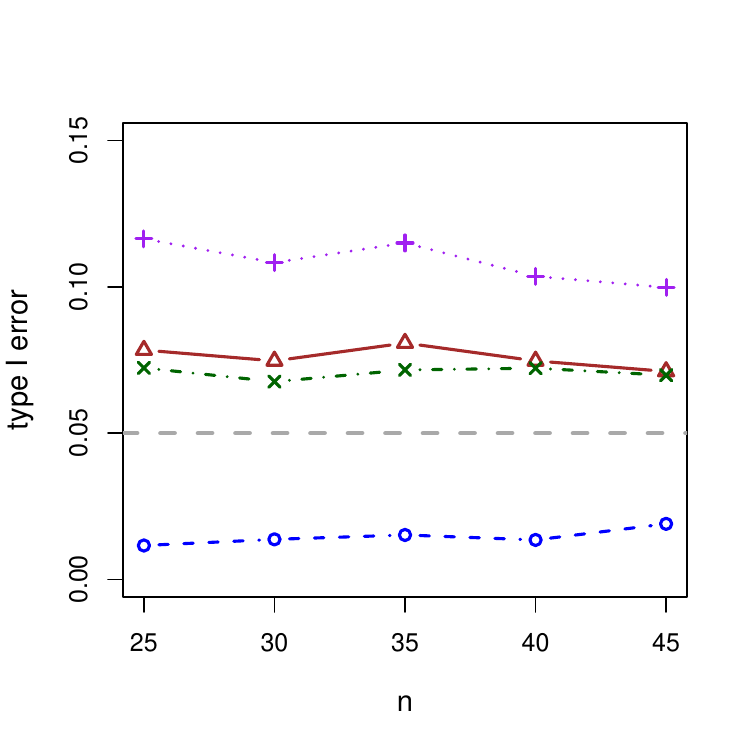}}
    \caption{Type~I errors based on 10,000 simulations at the 5\% significance level.}
    \label{fig:type1}
\end{figure}

Table \ref{tab:power} reports the exact rejection rates ($\%$) in Figure \ref{fig:power}.
\renewcommand{\arraystretch}{1.1}
\begin{table}[h!]
    \centering
    \begin{tabular}{c| p{1.5em}p{1.5em}p{1.5em}p{1.5em}p{1.5em}p{1.5em}p{1.5em}p{1.5em}p{1.5em}p{1.5em}p{1.5em}p{1.5em}p{1.5em}p{1.5em}p{1.5em}}
    \multicolumn{1}{c|}{$\beta$} & 0.01 & 0.02 & 0.03 & 0.04 & 0.05 & 0.06 & 0.07 & 0.08 & 0.09 & 0.10 & 0.11 & 0.12 & 0.13 & 0.14 & 0.15\\
    \hline
    \multicolumn{5}{l}{Panel A: $\phi_2 = 0.1$.} &  \multicolumn{5}{c}{}  \\
    \hline
    IPT & 8.8 & 27.6 & 49.5 & 67.3 & 79.7 & 87.9 & 93.3 & 96.0 & 97.7 & 98.6 & 99.2 & 99.6 & 99.8 & 99.9 & 99.9\\ 
    $t2$ & 14.2 & 36.7 & 58.0 & 74.4 & 84.8 & 91.1 & 94.8 & 97.0 & 98.5 & 99.1 & 99.5 & 99.7 & 99.8 & 99.9 & 99.9\\ 
    $t3$ & 31.4 & 54.6 & 73.8 & 85.3 & 92.0 & 95.4 & 97.7 & 98.8 & 99.4 & 99.7 & 99.8 & 99.8 & 99.9 & 100.0 & 100.0\\ 
    WCB & 10.6 & 23.6 & 39.3 & 55.1 & 67.8 & 77.4 & 84.5 & 89.2 & 93.0 & 95.3 & 96.6 & 97.6 & 98.5 & 99.1 & 99.4\\ 
    \hline
    \multicolumn{5}{l}{Panel B: $\phi_2 = 0.5$.} & \multicolumn{5}{c}{}  \\
    \hline
    IPT & 2.8 & 8.1 & 16.9 & 28.0 & 38.9 & 48.9 & 58.1 & 65.9 & 72.8 & 78.4 & 83.4 & 87.2 & 90.1 & 92.5 & 94.1\\ 
    $t2$ & 1.2 & 21.2 & 33.7 & 46.1 & 56.8 & 65.9 & 73.6 & 79.9 & 84.8 & 88.3 & 91.2 & 93.3 & 94.9 & 96.0 & 97.0\\ 
    $t3$ & 17.1 & 28.3 & 41.0 & 53.0 & 63.3 & 71.9 & 78.7 & 84.1 & 88.0 & 91.0 & 93.2 & 95.0 & 96.1 & 97.1 & 97.9\\ 
    WCB & 7.4 & 10.7 & 15.7 & 21.7 & 28.9 & 36.3 & 43.6 & 50.7 & 57.1 & 63.5 & 68.9 & 73.7 & 77.9 & 81.4 & 84.4\\ 
    \hline
    \multicolumn{5}{l}{Panel C: $\phi_2 = 0.9$.} & \multicolumn{5}{c}{}  \\
    \hline
    IPT & 1.2 & 1.9 & 3.0 & 4.9 & 6.8 & 9.5 & 12.6 & 16.2 & 19.5 & 23.5 & 27.7 & 31.7 & 35.5 & 39.3 & 43.0\\ 
    $t2$ & 8.2 & 10.5 & 13.7 & 17.5 & 21.4 & 26.2 & 30.8 & 35.4 & 40.0 & 44.9 & 49.1 & 52.9 & 56.9 & 60.3 & 63.5\\ 
    $t3$ & 12.1 & 14.8 & 18.0 & 21.9 & 26.2 & 30.9 & 35.4 & 40.0 & 44.8 & 49.2 & 53.3 & 57.1 & 60.6 & 63.9 & 66.9\\ 
    WCB & 8.4 & 8.6 & 9.1 & 10.2 & 11.3 & 12.7 & 14.4 & 16.3 & 18.4 & 20.7 & 23.3 & 25.8 & 28.4 & 31.4 & 34.3\\ 
    \end{tabular}
    \bigskip
    \caption{Test power based on 10,000 simulations at the 5\% significance level.}
    \label{tab:power}
\end{table}

\textbf{Heteroskedastic Errors.}
We illustrate this extension in a gravity-style dyadic regression. There are \(n=25\) exporters and \(n=25\) importers, with one observation per dyad. The outcome is generated from
\begin{equation}\label{eq:gravity_design}
    Y_{ij}
    =
    4
    +0.8\log(\mathrm{GDP}_i)
    +0.8\log(\mathrm{GDP}_j)
    + b \log(\mathrm{dist}_{ij})
    + \varepsilon_{ij}.
\end{equation}
The null hypothesis is \(H_0:b=0\). Exporter and importer GDPs are generated on an economically meaningful lognormal scale, while log distance is generated from a two-way random-effect covariate model. The heteroskedastic error is
\[
    \varepsilon_{ij}
    =
    \sigma_{ij}\eta_{ij},
    \qquad
    \sigma_{ij}
    =
    0.35\exp\{0.1 \widetilde\mu_{ij}+ 0.1 \widetilde D_{ij}\},
\]
In the error distribution above,  $\eta$ follows a symmetric normal distribution $\cN(0, 1)$, and \(\widetilde\mu_{ij}\) and \(\widetilde D_{ij}\) are standardized versions of the gravity mean and log distance, introducing heteroskedasticity. The sign-flip IPT uses sign-flipping groups constructed from Algorithm \ref{alg:sign} with $K = 2^6 - 1$.

Table~\ref{tab:sign} reports rejection rates as \(b\) varies from \(0\) to \(0.30\). At the null \(b=0\), IPT-Het rejects about \(1.4\%\) of the time, while the \(t_3\) statistic is oversized in this design, with rejection rates about \(8.5\%\). IPT-Het also has substantial power comparable to other tests. Therefore, our framework also yields power and finite-sample valid tests under heteroskedasticity, as long as double sign symmetry holds.

\begin{table}[h]
    \centering
    \begin{tabular}{c|ccccccc}
    \hline
    & \multicolumn{7}{c}{$b$} \\
    & 0.00 & 0.05 & 0.10 & 0.15 & 0.20 & 0.25 & 0.30 \\
    \hline
    IPT   & 1.40 & 14.05 & 56.09 & 86.26 & 95.42 & 97.98 & 98.88 \\
    $t_2$ & 0.72 & 13.95 & 69.76 & 97.62 & 99.99 & 100.00 & 100.00 \\
    $t_3$ & 8.49 & 41.77 & 90.26 & 99.67 & 100.00 & 100.00 & 100.00 \\
    WCB   & 6.09 & 34.43 & 84.51 & 99.02 & 99.99 & 100.00 & 100.00 \\
    \hline
    \end{tabular}
    \bigskip
    \caption{Rejection rates (\%) across different values of $b$.}
    \label{tab:sign}
\end{table}

\textbf{Pairwise IPT for Conservativeness.}
We evaluate the pairwise IPT under the gravity design \eqref{eq:gravity_design} with \(b=0\), exchangeable errors, and $1,000$ replications. We vary the number of exporters/importers over \(n\in\{20,22,24,26,28,30\}\). Figure~\ref{fig:new-sims} shows that the original IPT is indeed conservative in this design. The finite-sample valid pairwise IPT (using the \(\alpha/2\) correction) has a similar conservativeness issue to the original IPT. The naive pairwise IPT (using the test threshold $\alpha$) has an empirical type I error close to the nominal \(5\%\) level. Therefore, the conservativeness can be empirically addressed by a calibration change, without changing the main framework of the paper.

\begin{figure}[h]
\centering
\includegraphics[width=0.47\textwidth]{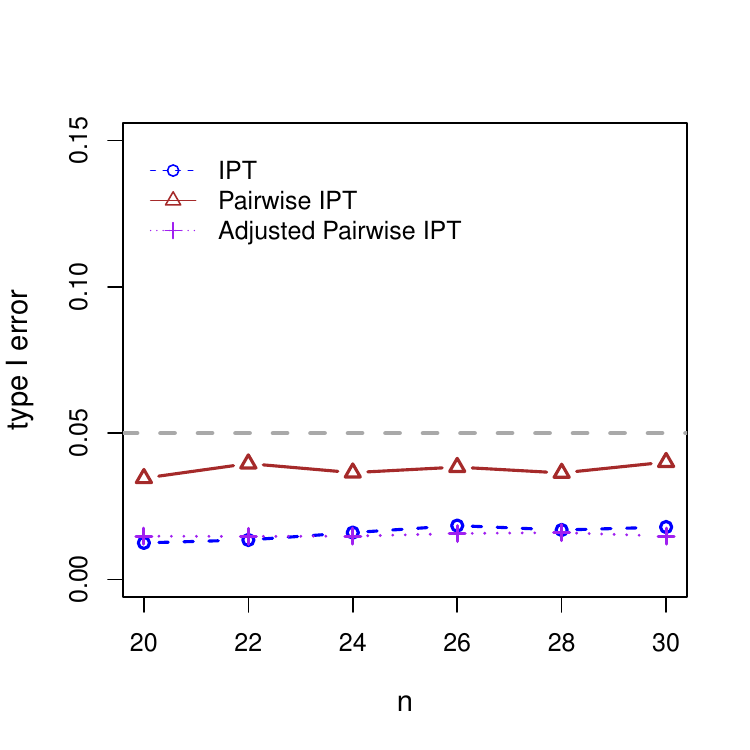}
\caption{Type-I error under \(b=0\) for the original IPT and the pairwise calibration. The dashed horizontal line marks the nominal \(5\%\) level.}
\label{fig:new-sims}
\end{figure}

\section{Additional Numerical Results in Section \ref{sec:trade}}\label{sec:trade_supp}
We revisit the empirical study of bilateral trade flows in Section \ref{sec:trade} and examine how the variance of each covariate influences the width of the confidence intervals. In standard regression settings, the sample variances of the covariates inform the widths. However, the dyadic regression in Section \ref{sec:trade} includes both dyad-level and node-level covariates, and their effective sample sizes differ substantially. To make meaningful comparisons, we define the effective variance of covariates
\[
\widehat{\text{EffVar}}_i = 
\begin{cases}
    n^2 \widehat{\mathrm{Var}}(X_i) &\text{if $X_i$ is a dyad-level covariate vector}\\
    n \widehat{\mathrm{Var}}(X_i) &\text{if $X_i$ is a node-level covariate vector}\;.
\end{cases}
\]
This scaling reflects the fact that dyad-level covariates contribute information at the order of $n^2$, whereas cluster-level covariates contribute information at the order of $n$.

Figure \ref{fig:ci_vs_var} plots the IPT confidence-interval width against the inverse of the effective variance, which serves as a proxy for Fisher information. We observe that covariates with smaller effective variance yield wider confidence intervals. The remoteness variables have the smallest effective variance and therefore produce the widest intervals.

We next summarize the variables identified as significant by each method in Table~\ref{tab:trade_variables}. The last two columns report results from the NLS and PPML estimates in \citet{silva2006log}. Compared with $t$-tests, our permutation test identifies more variables as non-significant --- such as importer’s per-capita GDP and openness dummy --- consistent with NLS and PPML. Our test also classifies exporter’s per-capita GDP as non-significant, whereas NLS and PPML find it significant. This discrepancy is plausible given the correlation ($\approx 0.6$) between exporter GDP and per-capita GDP, implying limited incremental explanatory power.
\begin{figure}[h!]
    \centering
    \includegraphics[width=1.0\linewidth]{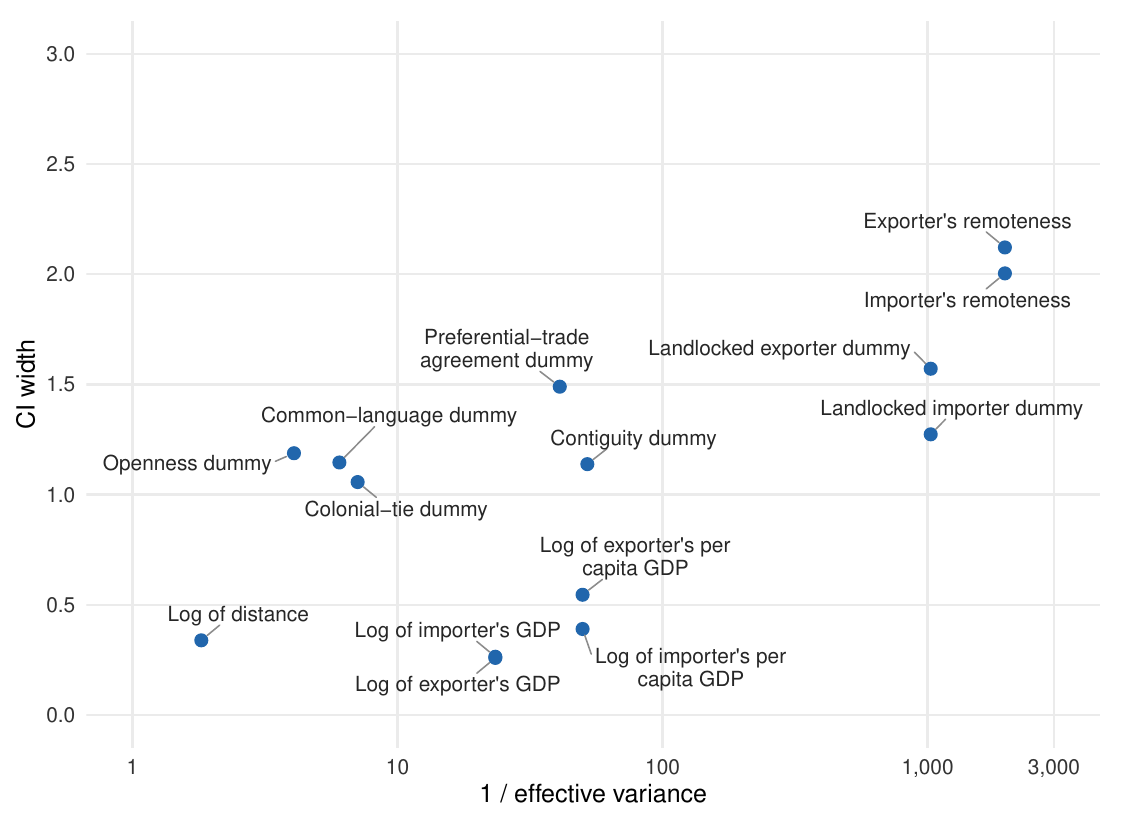}
    \caption{Widths of IPT confidence intervals over the inverse of effective variances.}
    \label{fig:ci_vs_var}
\end{figure}

\renewcommand{\arraystretch}{0.95}
\begin{table}[t!]
    \centering
    \begin{tabular}{l|cccccc}
    \hline
    Variables & $t_2$ & $t_3$ & WCB & IPT (ours) & NLS & PPML\\
    \hline
    \text{Log of importer's per capita GDP}  &$*$ &$*$ &- &$*$ & -& $*$\\
    \text{Log of exporter's per capita GDP}  &$*$ &$*$ &$*$ &- & $*$ & $*$\\
    \text{Log of importer's GDP} &$*$ &$*$ &$*$ &$*$ & $*$ & $*$\\
    \text{Log of exporter's GDP} &$*$ &$*$ &$*$ &$*$ & $*$ & $*$\\
    \text{Log of distance} &$*$ &$*$ &$*$ &$*$ & $*$ & $*$\\
    \text{Contiguity dummy} &- &- &- &- & - & - \\
    \text{Common-language dummy} &$*$ &$*$ &$*$ &$*$ & $*$ & $*$\\
    \text{Colonial-tie dummy} &$*$ &$*$ &- &- & - & -\\
    \text{Landlocked exporter dummy} &- &- &- &- & $*$& $*$\\
    \text{Landlocked importer dummy} &- &- &- &- & $*$& $*$\\
    \text{Importer's remoteness} &- &- &- &-& $*$& $*$\\
    \text{Exporter's remoteness} &- &- &- &- & $*$& $*$\\
    \text{Preferential-trade agreement dummy} &$*$ &$*$ &$*$ &$*$ & $*$& -\\
    \text{Openness dummy} &$*$ &$*$ &- &- & $*$& -\\
    \hline
    \end{tabular}
    \bigskip
    \caption{Variable significance detected by different tests. $*$ and “–” indicate significant and non-significant variables, respectively, at the 95\% level.}
    \label{tab:trade_variables}
\end{table}

\section{Finite-sample Validity of Procedure \ref{proc1}}\label{sec:proof_IPT}
First, we prove the validity of Procedure \ref{proc1} via the three steps in Section \ref{sec:overview}.

\begin{proof}[Proof of Theorem \ref{thm:valid}]
In Procedure \ref{proc1}, since $2p < N$, the orthonormal matrix $V_k$ is well-defined. Without loss of generality, we consider $d = 1$ such that $\DD$ is a length-$N$ vector. Under the null hypothesis $H_0: \beta=0$ and $V_k^\top \XX = 0$, we have $V_k^\top \yy = V_k^\top (\XX \gamma + \ee) = V_k^\top \ee$. Similarly, we have $V_k^\top \yy_{\pi_{k}, \sigma_{k}} = V_k^\top (\XX_{\pi_{k}, \sigma_{k}} \gamma + \ee_{\pi_{k}, \sigma_{k}}) = V_k^\top \ee_{\pi_k, \sigma_k}$. Therefore, by defining $f_k(v) = |\langle V_k^\top \mathbf{D}, V_k^\top v \rangle|$, we have
\[
a_k = |\langle V_k^\top \mathbf{D}, V_k^\top \ee \rangle| \eqqcolon f_k(\ee)\;, \quad b_k = |\langle V_k^\top \mathbf{D}, V_k^\top \ee_{\pi_k, \sigma_k} \rangle| \eqqcolon f_k(\ee_{\pi_k, \sigma_k})\;.
\]
For the randomization $p$-value, we have
\begin{align}
\mathrm{pval}
&\ge \frac{1}{K+1}\Bigl(1+\sum_{k=1}^K \mathbbm{1}\bigl\{\min _{1 \leq j \leq K} f_j(\ee) \le \min _{1 \leq j \leq K} f_j(\ee_{\pi_k, \sigma_k})\bigr\}\Bigr)\nonumber\\
&= \frac{1}{K+1}\Bigl(1+\sum_{k=1}^K \mathbbm{1}\bigl\{f^*(\ee) \le f^*(\ee_{\pi_k, \sigma_k})\bigr\}\Bigr)\;, \label{eq:pval1}
\end{align}
with $f^*(v) \coloneqq \min _{1 \leq j \leq K} f_j(v)$. We apply Lemma 3 of \citet{wen2025residual} to obtain
\begin{equation}\label{eq:pval2}
\P\Bigl(\frac{1}{K+1}\Bigl(1+\sum_{k=1}^K \mathbbm{1}\bigl\{f^*(\ee) \le f^*(\ee_{\pi_k, \sigma_k})\bigr\}\Bigr) \le \alpha\mid \XX, \DD \Bigr) \le \alpha\;.
\end{equation}
By combining Equations \eqref{eq:pval1} and \eqref{eq:pval2}, we have $\P(\mathrm{pval} \le \alpha\mid \XX, \DD) \le \alpha$.
\end{proof}

Next, we show that Algorithm \ref{alg:construct} returns an algebraic group for Procedure \ref{proc1}.
\begin{proposition}\label{prop:group}
Let $\{\psi_k\}_{k = 0}^K$ be the outputs from Algorithm \ref{alg:construct}. For any $r,s \in\{0,1, \ldots, K\}$, there exists $k\in\{0, \dots, K\}$ such that $\psi_{k}=\psi_{r} \circ \psi_{s}$, where $\circ$ denotes function composition.
\end{proposition}
\begin{proof}
For any $r, s$, we have $\psi_r \circ \psi_s = \pi^{-1} \circ \tilde{\psi}_r \circ \tilde{\psi}_s \circ \pi$. Let $k$ be the remainder after dividing $r+s$ by $K+1$. Then, one can verify $\tilde{\psi}_k = \tilde{\psi}_r \circ \tilde{\psi}_s$, therefore $\psi_{k}=\psi_{r} \circ \psi_{s}$.
\end{proof}

\section{Power Analysis of Procedure \ref{proc1}}\label{sec:proof_IPT_power}
We prove the power result in Section \ref{sec:power}, focusing on the $d = 1$ setting. Under Assumption \ref{asmp:err}, without loss of generality, we assume that 
\begin{equation}\label{eq:e_moments}
    1\le \E|\xi^e_{1}|^2 \le 2\;,\quad 1 \le \E|\eta^e_{1}|^2 < 2\;,\quad 1\le \E|\zeta^e_{11}|^2 \le 2\;.
\end{equation}
Let $\|\cdot\|$ denote the Euclidean norm for vectors and the (spectral) operator norm for matrices. Let $\|\cdot\|_F$ denote the Frobenius norm for matrices.

\subsection{Proof of Theorem \ref{thm:power}}
Lemmas \ref{lem:upp_bound} and \ref{lem:low_bound} establish concentration bounds for several key components. The probability statements are conditional on the covariates $\XX, \DD$, and we use $N = n^2$. Some supporting lemmas are provided in Section \ref{sec:support}.
\begin{lemma}\label{lem:upp_bound}
Suppose Assumptions \ref{asmp:model} and \ref{asmp:err} hold with $t\in[0, 1]$. For any $j$, we have $|{e^\top V_j V_j^\top \ee}| = o_p(n^{\frac{2+t}{1+t}})$ for $t\in[0, 1)$, and  $|{e^\top V_j V_j^\top \ee}| = O_p(n^{\frac{3}{2}})$ for $t = 1$.
\end{lemma}

\begin{proof}
\noindent\underline{Step 1. A deterministic vector $w$.}
First, we show that for any $w\in\R^{N}$,
\begin{align}
    |{w^\top V_j V_j^\top \ee}| &= o_p\left(\|w\| n^{\frac{1}{1+t}}\right) ~\text{for}~t\in[0, 1)\;, \label{eq:weps_bound}\\
    |{w^\top V_j V_j^\top \ee}| &= O_p \left(\|w\| n^{\frac{1}{2}}\right)~\text{for}~t=1\;.\label{eq:weps_bound2}
\end{align}
To simplify the notation, we connect $\ee$ to $\eta, \xi\in\R^n$ and $\zeta\in\R^{n^2}$ in the matrix form: 
\begin{align*}
    \ee &= A_1 \eta + A_2 \xi + \zeta\;, \quad
    A_1 = \begin{bmatrix}
        1_n & 0 &\dots & 0\\
        0 & 1_n & \dots & 0 \\
        \vdots & \vdots & \vdots & \vdots\\
        0 & \dots & 0 & 1_n 
    \end{bmatrix} \in \R^{n^2\times n}\;,\quad 
    A_2 = \begin{bmatrix}
        I_n \\
        I_n \\
        \vdots\\
        I_n 
    \end{bmatrix}\in\R^{n^2\times n}\;,
\end{align*}
where $1_n$ is an all-one vector of length $n$. Then, we can write 
\begin{align*}
    |{w^\top V_j V_j^\top \ee}| &= |w^\top V_j V_j^\top (A_1 \eta + A_2 \xi + \zeta)| \le \underbrace{|w^\top V_j V_j^\top A_1 \eta|}_\text{(I)} + \underbrace{|w^\top V_j V_j^\top  A_2 \xi|}_\text{(II)} + \underbrace{|w^\top V_j V_j^\top\zeta|}_\text{(III)}\;.
\end{align*}
\noindent\underline{Case 1: $t=1$. }
When $t = 1$, we directly analyze the second moment of (I) to obtain
\begin{align*}
    \E \text{(I)}^2 &= \E (w^\top V_j V_j^\top A_1 \eta)^2 \stackrel{\text{(i)}}{\le} 2 \tr(A_1^\top V_j V_j^\top w w^\top V_j V_j^\top A_1) \le 2 \|w\|^2 \|V_j V_j^\top\|^2 \|A_1\|^2\;,
\end{align*}
where (i) is due to \eqref{eq:e_moments}. Since $\|V_j V_j^\top\| \le 1$ and $\|A_1\|^2 = n$, by Chebyshev's inequality,
\begin{align*}
    \P(\text{(I)} > \|w\|n^{1/2}\delta) \le \frac{ \E \text{(I)}^2}{n \|w\|^2\delta^2} = O\left(\frac{1}{\delta^2}\right) \Rightarrow \text{(I)} = O_p(\|w\|n^{\frac{1}{2}})\;.
\end{align*}
Similarly, we can show that (II) is $O_p(\|w\|n^{\frac{1}{2}})$. Lastly, 
\begin{align*}
    \E \text{(III)}^2 &= \E (w^\top V_j V_j^\top \zeta)^2 \stackrel{\text{(i)}}{\le} 2 \tr(V_j V_j^\top w w^\top V_j V_j^\top) \le 2 \|w\|^2 \|V_j V_j^\top\|^2 \le 2 \|w\|^2\;,
\end{align*}
where inequality (i) is due to \eqref{eq:e_moments}. By Chebyshev's inequality, (III) is $O_p(\|w\|)$. Combining our analysis for (I) - (III), we prove Equation \eqref{eq:weps_bound2}.

\noindent\underline{Case 2: $t\in[0, 1)$. }
First we apply Corollary 8 and Lemma A8 of \citet{wen2025residual} to obtain
\begin{equation*}
    \text{(I)} = o_p\left(\|w^\top V_j V_j^\top A_1\| n^{\frac{1-t}{2(1+t)}}\right)\;.
\end{equation*}
Since $\|V_j V_j^\top\| \le 1$ and $\|A_1\| = \sqrt{n}$, we further obtain $\|w^\top V_j V_j^\top A_1\| \le \sqrt{n} \|w\|$ and hence $\text{(I)} = o_p\left(\|w\| n^{\frac{1}{2} + \frac{1-t}{2(1+t)}}\right)$. Using the same argument, we obtain the same bound for (II) and 
\begin{equation*}
    \text{(III)} = o_p\left(\|w^\top V_j V_j^\top\| (n^2)^{\frac{1-t}{2(1+t)}}\right) = o_p\left(\|w^\top V_j V_j^\top\| n^{\frac{1-t}{1+t}}\right)\;,
\end{equation*}
where we use the fact that $\zeta$ is an error vector of length $n^2$. Since $\|V_j V_j^\top\| \le 1$, we have $\text{(III)} = o_p\left(\|w\| n^{\frac{1-t}{1+t}}\right)$. Combining the analysis above gives us \eqref{eq:weps_bound}, i.e.,
\begin{gather}
    \P\left( |{w^\top V_j V_j^\top \ee}| \ge \delta \|w\| n^{\frac{1}{1+t}} \right) \to 0 ~\text{for any}~\delta > 0\;.\label{eq:w_opbound}
\end{gather}

\noindent\underline{Step 2. The random vector $e$.}
We replace $w$ by the random vector $e$ via a conditional argument. We focus on $t\in[0, 1)$, as the $t=1$ case follows a similar argument. Recall that Assumption \ref{asmp:model} implies $e = A_1 \eta^e + A_2 \xi^e + \zeta^e$. Therefore, 
\begin{equation*}
    |e^\top V_j V_j^\top \ee| \le |{\eta^e}^\top A_1^\top V_j V_j^\top \ee| + |{\xi^e}^\top A_2^\top V_j V_j^\top \ee| + |{\zeta^e}^\top V_j V_j^\top \ee|\;. 
\end{equation*}
Define the event
\begin{equation*}
    \cE \coloneqq \left\{ \frac{1}{2}n \le \|\eta^e\|^2 \le \frac{5}{2} n\right\} \cap \left\{ \frac{1}{2}n \le \|\xi^e\|^2 \le \frac{5}{2} n\right\} \cap \left\{ \frac{1}{2}n^2 \le \|\zeta^e\|^2 \le \frac{5}{2} n^2\right\} \;.
\end{equation*}
By the law of large numbers and Equation \eqref{eq:e_moments}, we have $\lim_{n\to\infty} \P(\cE) = 1$.
Hence, for the first term $|{\eta^e}^\top A_1^\top V_j V_j^\top \ee|$ above, we have
\begin{equation}\label{eq:sum_prob}
    \P(|{\eta^e}^\top A_1^\top V_j V_j^\top \ee| \ge \delta n^{\frac{2+t}{1+t}}) \le \P(|{\eta^e}^\top A_1^\top V_j V_j^\top \ee| \ge \delta  n^{\frac{2+t}{1+t}} \mid \cE) + 1 - \P(\cE)\;.
\end{equation}
Conditional on the event $\cE$, for $w = A_1 \eta^e$, we have
\begin{align}
    \|w\| &= \|A_1 \eta^e\| \le \sqrt{n} \|\eta^e\| \le \sqrt{\frac{5}{2}} n\;,\nonumber\\
    \P(|w^\top V_j V_j^\top \ee| \ge \delta n^{\frac{2+t}{1+t}} \mid \cE) &\le \P(|w^\top V_j V_j^\top \ee| \ge \delta \sqrt{\frac{2}{5}}\|w\|\times n^{\frac{1}{1+t}} \mid \cE) = o(1)\;.\label{eq:tail_prob}
\end{align}
The last line follows from \eqref{eq:w_opbound} and the independence between $e$ and $\ee$. Combining \eqref{eq:sum_prob} and \eqref{eq:tail_prob}, we obtain
\begin{align*}
\P(|{\eta^e}^\top A_1^\top V_j V_j^\top \ee| \ge \delta n^{\frac{2+t}{1+t}}) = o(1)\Rightarrow |{\eta^e}^\top A_1^\top V_j V_j^\top \ee| = o_p(n^{\frac{2+t}{1+t}})\;.
\end{align*}
Using a similar argument, we can show that $|{\xi^e}^\top A_2^\top V_j V_j^\top \ee|$ and $|{\zeta^e}^\top V_j V_j^\top \ee|$ are $o_p(n^{\frac{2+t}{1+t}})$. Therefore, we obtain $|e^\top V_j V_j^\top \ee| = o_p(n^{\frac{2+t}{1+t}})$.
\end{proof}

For any $k\in[K]$, let $P_k\in\R^{N\times N}$ be the permutation matrix such that  $\yy_{\pi_k, \sigma_k} = P_k \yy$.
\begin{lemma}\label{lem:low_bound}
Suppose that Assumptions \ref{asmp:model},\ref{asmp:err} hold and we have $n > (\max\{2/\kappa, 2\} +m)p$ and $\tr(P_k) = 0$ for any $k$. With probability converging to 1, for all $j, k\in[K]$, we have
\begin{align*}
\frac{e^\top V_j V_j^\top e - e^\top V_k V_k^\top P_k e}{n^2} & \geq \frac{\kappa^2 m}{2(2+\kappa m)} (\E(\xi_1^e)^2+\E(\eta_1^e)^2)\;, \\
\frac{e^\top V_j V_j^\top e + e^\top V_k V_k^\top P_k e}{n^2} & \geq \frac{\kappa^2 m}{2(2+\kappa m)} (\E(\xi_1^e)^2+\E(\eta_1^e)^2)\;.
\end{align*}
\end{lemma}

\begin{proof}
We focus on proving the first inequality, since the argument for the second inequality is identical. Since $K$ is fixed, it suffices to show for any fixed $j, k$, with probability converging to 1, the first statement holds. By the definition of $e$, we have
\begin{align*}
    \frac{e^\top V_j V_j^\top e - e^\top V_k V_k^\top P_k e}{n^2} &= \frac{(A_1 \eta^e + A_2 \xi^e + \zeta^e)^\top (V_j V_j^\top - V_k V_k^\top P_k) (A_1 \eta^e + A_2 \xi^e + \zeta^e)}{n^2}\;,
\end{align*}
where the matrices $A_1, A_2$ are defined in the proof of Lemma \ref{lem:upp_bound}. By expanding and rearranging the quadratic form above, we obtain
\begin{equation}\label{eq:quad_diff}
\begin{aligned}
&~~\frac{e^\top V_j V_j^\top e - e^\top V_k V_k^\top P_k e}{n^2} \\
&= \frac{(A_1 \eta^e)^\top (V_j V_j^\top - V_k V_k^\top P_k) (A_1 \eta^e)}{n^2} + \frac{(A_2 \xi^e)^\top (V_j V_j^\top - V_k V_k^\top P_k) (A_2 \xi^e)}{n^2} \\
&+ \frac{(\zeta^e)^\top (V_j V_j^\top - V_k V_k^\top P_k) (\zeta^e)}{n^2} \\
&+ \frac{(A_1 \eta^e)^\top (V_j V_j^\top - V_k V_k^\top P_k) (A_2 \xi^e)}{n^2} 
+ \frac{(A_2 \xi^e)^\top (V_j V_j^\top - V_k V_k^\top P_k) (A_1 \eta^e)}{n^2} \\
&+ \frac{(A_1 \eta^e)^\top (V_j V_j^\top - V_k V_k^\top P_k) (\zeta^e)}{n^2}
+ \frac{(\zeta^e)^\top (V_j V_j^\top - V_k V_k^\top P_k) (A_1 \eta^e)}{n^2} \\
&+ \frac{(A_2 \xi^e)^\top (V_j V_j^\top - V_k V_k^\top P_k) (\zeta^e)}{n^2}
+ \frac{(\zeta^e)^\top (V_j V_j^\top - V_k V_k^\top P_k) (A_2 \xi^e)}{n^2}\;.
\end{aligned}
\end{equation}
Next, we analyze the terms above one by one. For any square matrix $A$, $\mathrm{diag}(A)$ denotes the diagonal matrix whose diagonal entries equal those of $A$.

\noindent\underline{Term 1 in \eqref{eq:quad_diff}.}
For the first term in \eqref{eq:quad_diff}, we have
\begin{align*}
    \frac{(A_1 \eta^e)^\top (V_j V_j^\top - V_k V_k^\top P_k) (A_1 \eta^e)}{n^2}
    &= \frac{{\eta^e}^\top (M - \mathrm{diag}(M))\eta^e - {\eta^e}^\top (N - \mathrm{diag}(N)) \eta^e }{n^2} \\
    &+ \frac{{\eta^e}^\top (\mathrm{diag}(M) - \mathrm{diag}(N)) \eta^e}{n^2} \eqqcolon \text{I} +  \text{II}\;,
\end{align*}
where $M = A_1^\top V_j V_j^\top A_1$ and $N = A_1^\top V_k V_k^\top P_k A_1$. For I, observe that
\begin{align*}
    \|M - \diag(M)\|_F^2 &\le \|M\|_F^2 = \| A_1^\top V_j V_j^\top A_1 \|_F^2 \le \|V_j\|^4 \|A_1^\top A_1\|_F^2 = n^3\;,\\
    \|N - \diag(N)\|_F^2 &\le \|N\|_F^2 = \| A_1^\top V_j V_j^\top A_1 P_k \|_F^2 \le \|V_j\|^4 \|P_k\|^2 \|A_1^\top A_1\|_F^2 = n^3\;.
\end{align*}
In the derivation, we use the fact that $\|V_j\| \le 1$, $\|P_k\|\le 1$, $\|A_1^\top A_1\|_F^2 = n^3$. We apply \eqref{eq:sym_quad_form} in Lemma \ref{lem:quad_form} to obtain that for any constant $\delta > 0$, $\lim_{n\to\infty} \P(|\text{I}| < \delta) = 1$.

For II, given any fixed $P_k, P_j$, we define $a_{n,i} = (M_{ii}-N_{ii})/n^2$ and write $V_{n,i}\coloneqq n a_{n,i} (\eta^e_i)^2$. Then we can rewrite II as $\text{II} = \frac{1}{n}\sum_{i=1}^n V_{n,i}$. Notice that for each $n$, it holds that
\begin{align*}
|M_{ii}| &= A_{1,i}^\top V_j V_j^\top A_{1,i} \le \|A_{1,i}\|^2 \| V_j V_j^\top\| = n\;,\\
|N_{ii}| &= A_{1,i}^\top V_k V_k^\top P_k A_{1,i} \le \|A_{1,i}\|^2 \| V_k V_k^\top P_k\| = n \| V_k V_k^\top\| \le n\;,\\
\Rightarrow~ |a_{n, i}| &\le \frac{1}{n^2} (|M_{ii}| + |N_{ii}|) \le \frac{2}{n}\;.
\end{align*}
From this, we have $\E|V_{n, i}| \le 2 \E (\eta_1^e)^2$ uniformly for all $i, n$, and that for any $a>0$, 
\begin{equation*}
\sup _{n \geq 1} \frac{1}{n} \sum_{i=1}^n \E\left[|V_{n,i}| \mathbb{I}\{|V_{n, i}|>a\}\right] \le \mathbb{E}\left(2 (\eta_1^e)^2 \mathbb{I}\{2 (\eta_1^e)^2>a\}\right)\;.
\end{equation*}
Using the dominated convergence theorem and that $\E(\eta_1^e)^2<\infty$, we have that $\mathbb{E}(2 (\eta_1^e)^2 \mathbb{I}\{2 (\eta_1^e)^2>a\}) \to 0$ as $a \to \infty$; then we apply Lemma A7 of \citet{wen2025residual} to show that II converges to $\E[\mathrm{II}]$ in probability. Then it remains to control $\E \text{II} = \E (\eta_1^e)^2 \sum_i a_{n,i}$. Observe that $\sum_i a_{n,i} = \frac{1}{n^2} \tr(M - N)$. For matrix $M$, we have
\begin{align*}
    \tr(M) = \tr(A_1^\top V_j V_j^\top A_1) = n \tr(\frac{1}{n}A_1 A_1^\top V_j V_j^\top)\;.
\end{align*}
One can verify by definition that $\frac{1}{n} A_1 A_1^\top$ and $V_j V_j^\top$ are two $n^2\times n^2$ projection matrices with rank $n$ and $n^2 - 2p$, hence we apply Equation \eqref{eq:proj1} of Lemma \ref{lem:proj} to obtain
\begin{align*}
\tr(\frac{1}{n}A_1 A_1^\top V_j V_j^\top) \ge \max\{0, n + n^2 - 2p - n^2\} = \max\{0, n - 2p\} \Rightarrow \tr(M) \ge n \max\{0, n - 2p\} \;.
\end{align*}
For the matrix $N$, we have $|\tr(N)| = n \Bigl|\tr(\frac{1}{n} A_1 A_1^\top V_k V_k^\top P_k)\Bigr|$.
Similarly, we apply \eqref{eq:proj2} of Lemma \ref{lem:proj} to obtain $|\tr(N)| \le n \min\{n, n^2 - 2p\}$. Therefore, 
\[
\sum_{i=1}^n a_{n,i} = \frac{1}{n^2} \tr(M - N) \ge \frac{\max\{0, n - 2p\} - \min\{n, n^2 - 2p\}}{n}\;.
\]
Under the assumption $n > (\max\{2/\kappa, 2\} + m) p$, we have $n > 2p$ and $\sum_{i=1}^n a_{n,i} \ge \frac{ n - 2p - n}{n} = -\frac{2p}{n}$. For any constant $c>0$, with probability converging to 1, we have
\begin{equation*}
    \frac{(A_1 \eta^e)^\top (V_j V_j^\top - V_k V_k^\top P_k) (A_1 \eta^e)}{n^2} \ge -\frac{2p}{n} \E(\eta_1^e)^2 - c\;.
\end{equation*}

\noindent\underline{Term 2 in \eqref{eq:quad_diff}.}
By a similar analysis as above, with probability converging to 1, the term 2 is lower bounded by $-\frac{2p}{n} \E(\xi_1^e)^2 -c$.

\noindent\underline{Term 3 in \eqref{eq:quad_diff}.}
We apply a similar analysis to Term 3 and show that it concentrates around 
\begin{equation*}
    \frac{1}{n^2} \E (\zeta^e_{11})^2 \tr(V_j V_j^\top - V_k V_k^\top P_k) \ge \frac{1}{n^2} \tr(V_j V_j^\top - V_k V_k^\top P_k)\;,
\end{equation*}
where the last inequality follows from \eqref{eq:e_moments}. By the construction of $V_j$, we have $\tr (V_j V_j^\top) = n^2 - 2p$. Additionally, since $\tr(P_k) = 0$,
\begin{equation*}
    \Bigl|\tr(V_k V_k^\top P_k)\Bigr| = \Bigl|\tr((I_{n^2} - V_k V_k^\top) P_k)\Bigr| \le \tr(I_{n^2} - V_k V_k^\top) = 2p\;,
\end{equation*}
where the inequality follows from Lemma \ref{lem:MP-ineq}. Hence, $\tr(V_j V_j^\top - V_k V_k^\top P_k) \ge n^2- 4p$. Therefore, for any constant $c > 0$, with probability converging to 1, we have
\begin{equation*}
    \frac{(\zeta^e)^\top (V_j V_j^\top - V_k V_k^\top P_k) (\zeta^e)}{n^2} \ge \frac{n^2 - 4p}{n^2} \E (\zeta_{11}^e)^2 - c\;.
\end{equation*}

\noindent\underline{Remaining Terms in \eqref{eq:quad_diff}.}
We show that the rest of terms of Equation \eqref{eq:quad_diff} are $o_p(1)$. First, for 
\begin{equation*}
    \frac{(A_1 \eta^e)^\top (V_j V_j^\top - V_k V_k^\top P_k) (A_2 \xi^e)}{n^2} = \frac{(\eta^e)^\top (A_1^\top V_j V_j^\top A_2) (\xi^e)}{n^2} - \frac{(\eta^e)^\top (A_1^\top V_k V_k^\top P_k A_2) (\xi^e)}{n^2}\;,
\end{equation*}
we have
\begin{align*}
    \|A_1^\top V_j V_j^\top A_2\|_F^2 &\le \|A_1\|^2 \|V_j\|^4 \|A_2\|_F^2 \le n \|A_2\|_F^2 = n^3\;,\\
    \|A_1^\top V_k V_k^\top P_k A_2\|_F^2 &\le \|A_1\|^2 \|V_k\|^4 \|P_k\| \|A_2\|_F^2 \le n \| A_2\|_F^2 = n^3\;,
\end{align*}
where we use the fact that $\|A_2\|_F^2 = n^2$ from direct calculation. Therefore, we can apply Equation \eqref{eq:asym_quad_form} in Lemma \ref{lem:quad_form} to show that this term is $o_p(1)$. Using a similar argument, we can show the rest of terms are $o_p(1)$ as well.

Combining our analysis for terms 1-9 on the right hand side of \eqref{eq:quad_diff}, we have shown that with probability converging to 1, 
\begin{equation*}
    \frac{e^\top V_j V_j^\top e - e^\top V_k V_k^\top P_k e}{n^2} \ge  \frac{n^2 - 4p}{n^2} \E (\zeta_{11}^e)^2  -\frac{2p}{n} (\E(\xi_1^e)^2+\E(\eta_1^e)^2) -3 c
\end{equation*}
for any constant $c>0$. Under the condition that $\E (\zeta_{11}^e)^2 \ge \kappa (\E(\eta_1^e)^2 + \E(\xi_1^e)^2)$ and $n > (2 / \kappa + m) p$, we have
\begin{align*}
    & \frac{e^\top V_j V_j^\top e - e^\top V_k V_k^\top P_k e}{n^2} \ge  \Bigl(\frac{n^2 - 4p}{n^2} \kappa -\frac{2p}{n} \Bigr) (\E(\xi_1^e)^2+\E(\eta_1^e)^2) -3 c\\
    &\ge \Bigl( \kappa - \frac{4p\kappa}{n^2} - \frac{2p}{(2/\kappa + m)p}\Bigr) (\E(\xi_1^e)^2+\E(\eta_1^e)^2) -3 c = \Bigl( \frac{\kappa^2 m}{2 + \kappa m} - \frac{4p \kappa}{n^2} \Bigr)(\E(\xi_1^e)^2+\E(\eta_1^e)^2) -3 c\;.
\end{align*}
Note that $c$ can be taken arbitrarily small and the term $4p\kappa / n^2$ converges to zero as $n > 2p$. Therefore, with probability converging to one, we have
\[
\frac{e^\top V_j V_j^\top e - e^\top V_k V_k^\top P_k e}{n^2} \ge \frac{\kappa^2 m}{2 (2 + \kappa m)}(\E(\xi_1^e)^2+\E(\eta_1^e)^2)\;.
\]
\end{proof}

Next, we prove Theorem \ref{thm:power} based on Lemmas \ref{lem:upp_bound} and \ref{lem:low_bound}.
\begin{proof}
By the proof of \citet[Theorem 3]{wen2025residual}, it suffices to show for all $\delta>0$,
\begin{equation}\label{eq:upp_bound}
\begin{gathered}
\P\left(\exists 1 \le k \le K~\text{s.t.}~\frac{\left|{e}^{\top} V_k V_k^\top \ee \right|}{b n^2} \ge \delta \mid \XX, \DD \right) \to 0\;, \\
\P\left(\exists 1 \le k \le K~\text{s.t.}~\frac{\left|{e}^{\top} V_k V_k^\top P_k \ee \right|}{b n^2} \ge \delta\mid \XX, \DD \right) \to 0\;,
\end{gathered}
\end{equation}
and that with probability converging to 1, for all $1 \le j, k, \le K$,
\begin{equation}\label{eq:low_bound}
\begin{aligned}
\frac{e^\top V_j V_j^\top e - e^\top V_k V_k^\top P_k e}{n^2} & \geq \frac{\kappa^2 m}{2(2+\kappa m)} (\E(\xi_1^e)^2+\E(\eta_1^e)^2)\;, \\
\frac{e^\top V_j V_j^\top e + e^\top V_k V_k^\top P_k e}{n^2} & \geq \frac{\kappa^2 m}{2(2+\kappa m)} (\E(\xi_1^e)^2+\E(\eta_1^e)^2)\;.
\end{aligned}
\end{equation}

Since $n$ can be divided by $K+1$, by Algorithm \ref{alg:construct}, one can verify that for any $i, j \in [n]$ and any $k = 1, \dots, K$, $\pi_k(i) \neq i$, $\sigma_k(j) \neq j$. This implies that the corresponding permutation matrix $P_k$ satisfies $\tr(P_k) = 0$. By Assumptions \ref{asmp:model} and \ref{asmp:err}, we apply Lemma \ref{lem:low_bound} to obtain \eqref{eq:low_bound}.

Then it suffices to prove the first inequality in \eqref{eq:upp_bound}, as the second one can be proved similarly based on $P_k\ee \stackrel{d}{=} \ee$. By the union bound, we obtain
\begin{equation*}
\P\left(\exists 1 \le k \le K~\text{s.t.}~\frac{\left|{e}^{\top} V_k V_k^\top \ee \right|}{b n^2} \ge \delta \mid \XX, \DD\right) \le \sum_{k=1}^K \P\left(\frac{\left|{e}^{\top} V_k V_k^\top \ee \right|}{b n^2} \ge \delta \mid \XX, \DD \right)\;.
\end{equation*}
When $t = 1$, $bn^2 = \omega(n^{3/2})$. By Lemma \ref{lem:upp_bound}, $|{e}^{\top} V_k V_k^\top \ee| = O_p(n^{3/2})$ for any $k$, and hence $|{e}^{\top} V_k V_k^\top \ee|/bn^2 = o_p(1)$. Since $K$ is fixed, the sum of probabilities is $o(1)$.

For $t\in[0, 1)$, we have $bn^2 = \Omega(n^{\frac{2+t}{1+t}})$. By Lemma \ref{lem:upp_bound}, $|{e}^{\top} V_k V_k^\top \ee| = o_p(n^{\frac{2+t}{1+t}})$, and hence $|{e}^{\top} V_k V_k^\top \ee|/bn^2 = o_p(1)$. The sum of probabilities above is also $o(1)$. 
\end{proof}

\subsection{Proof of Theorem \ref{thm:power_iid}}
Under the stronger assumption that $\eta = \xi = 0$, the results in Lemma \ref{lem:upp_bound} can be strengthened in the following sense. We omit the proof as it follows from Lemma \ref{lem:upp_bound} and $\eta = \xi = 0$.
\begin{lemma}\label{lem:upp_bound1}
Suppose Assumptions \ref{asmp:model} and \ref{asmp:err} hold with $t\in[0, 1]$ and $\eta = \xi = 0$. For any $j$, we have
\begin{equation*}
    |{e^\top V_j V_j^\top \ee}| = o_p\left(n^{\frac{2}{1+t}}\right)~\text{for}~t\in[0, 1)\;,\quad 
    |{e^\top V_j V_j^\top \ee}| = O_p\left(n\right)~\text{for}~t = 1\;.
\end{equation*}
\end{lemma}
\begin{proof}[Proof of Theorem \ref{thm:power_iid}]
The proof follows from the same line of argument as in Theorem \ref{thm:power}, while we replace Lemma \ref{lem:upp_bound} with Lemma \ref{lem:upp_bound1} to get the improved rate.
\end{proof}

\section{Finite-sample Validity of Procedure \ref{proc:CIPT}}\label{sec:proof_CIPT}
\begin{proof}[Proof of Theorem \ref{thm:cond_valid}]
By construction in Procedure~\ref{proc:CIPT}, the permutations are applied to the collection $\{(\XX_q, \DD_q, \yy_q)\}$ and are restricted to sub-array permutations on $I_q\times J_q$ for each $q = 1, \dots, Q$. Hence, all permutations are confined to observed entries, and the same validity argument as in Section~\ref{sec:proof_IPT} applies.
\end{proof}

\section{Power Analysis of Procedure \ref{proc:CIPT}}\label{sec:proof_CIPT_power}

\subsection{Proof of Proposition \ref{thm:biclique}}
Given a random bipartite graph $G$ in Assumption \ref{A:random-graph}, let $X_s$ denote the number of bicliques of size $s\times s$. Recall that $\rodds = \rho / (1-\rho)$. Lemma \ref{lem:graph_low} shows that $X_s > 0$ with high probability, therefore a biclique of size $s\times s$ exists.

\begin{lemma}\label{lem:graph_low}
Let $s = \min\{(2/3 - b_0) \rodds \log n, a_n {n}^{1/3}\}$ for any constant $b_0 > 0$ and a diminishing sequence $a_n = o(1)$. We have $\lim_{n\to\infty}\P(X_s > 0) = 1$.
\end{lemma}
\begin{proof}[Proof of Lemma \ref{lem:graph_low}]
By definition, $X_s = \sum_{S, S'} \mathbb{I}\{(S, S')~\text{is a biclique}\}$, where $S$ and $S'$ are subsets of $[n]$ with size $s$. To analyze the tail behavior of $X_s$, let $S, S', T, T'$ be subsets of $[n]$ with size $s$, and denote $S \times S' \sim T \times T'$ if $S \times S' \cap T \times T' \neq \emptyset$. Then, define
$$
\bar{\Delta}=\sum_{\substack{(S,S') \sim (T, T')}} \mathbb{P}\left((S, S'), (T, T') \text { are $s\times s$ bicliques in } G\right)\;.
$$
Then, by Janson's inequality (Theorem 27.13 of \citet{frieze2015introduction}) we have
$$
\mathbb{P}\left(X_s=0\right) \le \exp(-\frac{\left(\mathbb{E} X_s\right)^2}{2 \bar{\Delta}})\;.
$$
It suffices to show that the ratio $\bar{\Delta}/(\E X_s)^2$ goes to zero. By definition of $X_s$, we have 
\begin{align*}
    \E X_s & = \sum_{S, S'} \P((S, S')~\text{is a biclique}) \stackrel{\text{(i)}}{=} \sum_{S, S'} \rho^{s^2} \stackrel{\text{(ii)}}{=} \binom{n}{s}^2 \rho^{s^2}\;.
\end{align*}
Here (i) holds because $(S, S')$ has $s^2$ edges in total and each edge occurs with probability $\rho$. (ii) follows from the fact that there are $\binom{n}{s}$ different choices for $S$ and $S'$, and hence we have $\binom{n}{s}^2$ combinations of $(S, S')$. Similarly, to compute $\bar{\Delta}$, we first enumerate all possible pairs of $(S, S')$, then enumerate $(T, T')$ where the overlapping indices vary from $1$ to $s$.
\begin{align*}
    \bar{\Delta} &= \sum_{\substack{(S,S') \sim (T, T')}} \mathbb{P}\left((S, S'), (T, T') \text { are $s\times s$ bicliques}\right)\\
    &= \sum_{S, S'} \sum_{T, T'} \mathbb{I}\{|S\cap T| \ge 1\} \mathbb{I}\{|S'\cap T'| \ge 1\} \P\left((S, S'), (T, T') \text { are $s\times s$ bicliques}\right) \\
    &= \sum_{S, S'} \sum_{T, T'} \sum_{i=1}^s \sum_{j=1}^s \mathbb{I}\{|S\cap T| = i\} \mathbb{I}\{|S'\cap T'| = j\} \P\left((S, S'), (T, T') \text { are $s\times s$ bicliques}\right)\;.
\end{align*}
For each pair of $i$ and $j$, it requires $s^2$ edges for $(S, S')$ to be a biclique, and $(s^2 - ij)$ additional edges for $(T, T')$ to be a biclique. Therefore, $\P((S, S'), (T, T') \text { are $s\times s$ bicliques}) = \rho^{s^2} \times \rho^{s^2 - ij}$ and
\begin{align*}
    \bar{\Delta} &=\rho^{s^2} \sum_{S, S'} \sum_{T, T'} \sum_{i=1}^s \sum_{j=1}^s \mathbb{I}\{|S\cap T| = i\} \mathbb{I}\{|S'\cap T'| = j\} \rho^{s^2 - ij} \\
    &= \rho^{s^2} \sum_{S, S'} \sum_{i=1}^s \sum_{j=1}^s \binom{n-s}{s-i}\binom{s}{i} \binom{n-s}{s-j}\binom{s}{j} \rho^{s^2 - ij}\;.
\end{align*}
In the last equality, the number of choices for $T$ is $\binom{n-s}{s-i}\binom{s}{i}$ since we select $i$ indices from $S$ and $s - i$ indices from the complement of $S$. Similarly, the number of choices for $T'$ is $\binom{n-s}{s-j}\binom{s}{j}$. Lastly, as shown earlier, there are $\binom{n}{s}$ choices for $S$ and $S'$, which leads to 
\[
\bar{\Delta} = \binom{n}{s}^2 \rho^{s^2} \sum_{i=1}^s \sum_{j=1}^s \binom{n-s}{s-i}\binom{s}{i} \binom{n-s}{s-j}\binom{s}{j}\rho^{s^2-ij}\;.
\]

Based on the expression of $\E X_s$ and $\bar{\Delta}$, we have
\begin{align*}
\frac{\bar{\Delta}}{\left(\mathbb{E} X_s\right)^2} & =\frac{\binom{n}{s}^2 \rho^{s^2} \sum_{i=1}^s \sum_{j=1}^s \binom{n-s}{s-i}\binom{s}{i} \binom{n-s}{s-j}\binom{s}{j}\rho^{s^2-ij}}{\left(\binom{n}{s}^2 \rho^{s^2}\right)^2} = \sum_{i=1}^s \sum_{j=1}^s \frac{\binom{n-s}{s-i}\binom{s}{i} \binom{n-s}{s-j}\binom{s}{j}}{\binom{n}{s}^2} \rho^{-ij} \eqqcolon \sum_{i,j=1}^s u_{i,j}.
\end{align*}
In the following we show that $u_{i,j}\le u_{1,1}$. By definition, for any $i,j$,
\begin{align*}
\frac{u_{i, j+1}}{u_{i,j}} &= 
\frac{(s-j)^2}{(n-2s + j + 1)(j+1)} \rho^{-i} \le \frac{s^2}{(n-2s)(j+1)} \rho^{-i} = \left(1 + \frac{2s}{n-2s}\right) \frac{s^2}{n (j+1)} \rho^{-i}\;.
\end{align*}
Similarly, 
\begin{equation*}
    \frac{u_{i+1, j}}{u_{i,j}} \le \left(1 + \frac{2s}{n-2s}\right) \frac{s^2}{n (i+1)} \rho^{-j}\;.
\end{equation*}
Therefore, for any $i, j$,
\begin{align*}
    \frac{u_{i,j}}{u_{1,1}} &\le \left(1 + \frac{2s}{n-2s}\right)^{i+j} \left(\frac{s^2}{n}\right)^{i+j} \frac{\rho^{-ij-1}}{i! j!}\\
    &\stackrel{\text{(i)}}{\le} \left(1 + \frac{2s}{n-2s}\right)^{i+j} \left(\frac{es^2}{in}\right)^{i} \left(\frac{es^2}{jn}\right)^{j} \rho^{-ij-1} \\
    &\stackrel{\text{(ii)}}{\le} 
    \left(1 + \frac{2s}{n-2s}\right)^{i+j} \left(\frac{es^2}{in}\cdot \rho^{-i/2}\right)^{i} \left(\frac{es^2}{jn}\cdot \rho^{-j/2}\right)^{j} \rho^{-1}
\end{align*}
where (i), (ii) follow from $i! \ge (i/e)^i$ and $ij\le (i^2 + j^2)/2$, respectively. Since
\begin{align*}
\frac{es^2}{in}\cdot \rho^{-i/2} = \exp(1 + 2 \log s - \log i - \log n + \frac{i}{2} \log \frac{1}{\rho})
\end{align*}
and $1/\rho = 1 + 1/\rodds$ and $\log(1+x) \asymp x$, we have $\log \frac{1}{\rho} \asymp \frac{1}{\rodds}$ and
\begin{align*}
\frac{es^2}{in}\cdot \rho^{-i/2} &\le \exp(2 \log s - \log n + \frac{s}{2\rodds}) \le \exp(2 \log s - \log n + \frac{(2/3-b_0)\rodds\log n}{2\rodds})\\
&= \exp(2 \log s - \log n + (1/3-b_0/2)\log n)\\
&\le \exp(2/3 \log n - \log n + 1/3 \log n  - b_0 \log n / 2) = \exp(- b_0 \log n / 2) = o(1)\;,
\end{align*}
where the last inequality follows from $s \le n^{1/3}$. Similarly, we can show that $\frac{es^2}{jn}\cdot \rho^{-j/2} = o(1)$. In addition, by the inequality $1+x \le \exp(x)$, we have
\begin{equation*}
    \left(1 + \frac{2s}{n-2s}\right)^{i+j} \le \exp(\frac{2s}{n-2s} (i+j)) \le \exp(\frac{4s^2}{n-2s}) = O(1)\;,
\end{equation*}
where the last equality follows from $s \le n^{1/3}$. Therefore, for $n$ large enough, we have $u_{i,j}/u_{1,1}<1$, i.e., $0 < u_{i,j} < u_{1,1}$.

Based on the analysis above, we have
\begin{align*}
\frac{\bar{\Delta}}{\left(\mathbb{E} X_s\right)^2} &\le \sum_{i,j=1}^s u_{i,j} \le s^2 u_{1,1} = s^2 \frac{\binom{n-s}{s-1}\binom{s}{1} \binom{n-s}{s-1}\binom{s}{1}}{\binom{n}{s}^2} \rho^{-1} = O\Bigl(\frac{s^6}{n^2}\Bigr)\;.
\end{align*}
Since $s = o(n^{1/3})$, the above quantity is $o(1)$, which completes the proof.
\end{proof}

\begin{proof}[Proof of Proposition \ref{thm:biclique}]
By setting $b_0 = 1/3$ and $a_n = 1/\log n$ in the parameter $s$ in Lemma \ref{lem:graph_low}, we have that with probability converging to 1, there exists a biclique of size $\min\Bigl\{\frac{1}{3} \rodds \log n, \frac{n^{1/3}}{\log n}\Bigr\}$.
\end{proof}

\subsection{Proof of Theorem \ref{thm:power_miss}}
Suppose the largest biclique for the permutation block is of size $s_1\times s_1$, and define the event 
\begin{equation*}
    \cE = \{ s_1 \ge s \}\;, \quad s = \min\Bigl\{\frac{1}{3} \rodds \log n, \frac{n^{1/3}}{\log n}\Bigr\}\;.
\end{equation*}
Based on Proposition \ref{thm:biclique}, we have $\P(\cE) \to 1$ as $n$ goes to infinity. Since the randomness of $\cE$ only comes from the random bipartite graph, it is independent of covariates and thus we have the conditional convergence $\P(\cE \mid \{x_{ij}, d_{ij}\}_{i,j\in[n]}) \to 1$.

To establish the diminishing type II error, note that 
\begin{align*}
\P\Bigl(\mathrm{pval} > \frac{1}{K+1} \mid \{x_{ij}, d_{ij}\}_{i,j\in[n]} \Bigr) &= \P\Bigl(\mathrm{pval} > \frac{1}{K+1} \mid \cE, \{x_{ij}, d_{ij}\}_{i,j\in[n]}\Bigr) \P(\cE \mid \{x_{ij}, d_{ij}\}_{i,j\in[n]}) \\
&+ \P\Bigl(\mathrm{pval} > \frac{1}{K+1} \mid \cE^\complement, \{x_{ij}, d_{ij}\}_{i,j\in[n]}\Bigr) \P(\cE^\complement \mid \{x_{ij}, d_{ij}\}_{i,j\in[n]}) \\
&\le \P\Bigl(\mathrm{pval} > \frac{1}{K+1} \mid \cE, \{x_{ij}, d_{ij}\}_{i,j\in[n]} \Bigr) + \P(\cE^\complement \mid \{x_{ij}, d_{ij}\}_{i,j\in[n]})\;.
\end{align*}
Since $\P(\cE^\complement \mid \{x_{ij}, d_{ij}\}_{i,j\in[n]})$ is $o(1)$, it suffices to ensure that $\P(\mathrm{pval} > \frac{1}{K+1} \mid \cE, \{x_{ij}, d_{ij}\}_{i,j\in[n]})$ is also $o(1)$. Since the test is conditioned on the single permutation block and remove the other observations, we may write
\begin{align*}
\P\Bigl(\mathrm{pval} > \frac{1}{K+1} \mid \cE, \{x_{ij}, d_{ij}\}_{i,j\in[n]}\Bigr) 
&= \P\Bigl(\mathrm{pval} > \frac{1}{K+1} \mid \{x_{ij}, d_{ij}\}_{i\in I, j\in J}, ~|I| = |J| \ge s \Bigr)\;,
\end{align*}
where we denote by $(I, J)$ the largest biclique of the random bipartite graph. For each fixed $(I, J)$ that satisfies $|I| = |J| \ge s$, we apply Theorem \ref{thm:power} with the number of clusters $n = s$, which shows that it converges to zero. Therefore, the probability above converges to zero.

\addtocounter{lemma}{4}
\section{Supporting Lemmas}\label{sec:support}
We introduce the following results adapted from Lemma A13 and Lemma A6 from \citet{wen2025residual}, respectively. We also provide the proofs for completeness. 
\begin{lemma}\label{lem:MP-ineq}
Consider a symmetric positive semidefinite matrix $M \in \R^{n \times n}$ and a permutation matrix $P\in \R^{n \times n}$, we have $|\tr(MP)| \le \tr(M)$.
\end{lemma}
\begin{proof}
Using the positive semidefiniteness and symmetry of $M$, we have for any $i, j$, $|{M}_{ij}| \le \frac{{M}_{ii}+{M}_{jj}}{2}$. Let $\sigma$ be the permutation induced by $P$. We have
$$
\tr(MP) = \sum_{i=1}^n M_{i, \sigma(i)} \le \sum_{i=1}^n |M_{i, \sigma(i)}| \le \sum_{i=1}^n \frac{M_{ii}+{M}_{\sigma(i) \sigma(i)}}{2}=\tr(M)\;.
$$
\end{proof}

\begin{lemma}\label{lem:quad_form}
Let $u \in \R^m$ be a random vector with i.i.d. entries and $\E u_1 = 0$, $\E u_1^2 \le 2$. For a symmetric matrix $M\in\R^{m\times m}$ with zero diagonal entries, we have for any fixed $\delta > 0$, 
\begin{equation}\label{eq:sym_quad_form}
    \P\left( \frac{u^\top M u}{n^2} > \delta \right) \le \frac{8 \|M\|_F^2}{n^4 \delta^2}\;.
\end{equation}
Let $v \in \R^k$ be another random vector with i.i.d. entries and $\E v_1 = 0$, $\E v_1^2 \le 2$, $u\indep v$. For any matrix $N\in\R^{m\times k}$, we have for any fixed $\delta > 0$, 
\begin{equation}\label{eq:asym_quad_form}
    \P\left( \frac{u^\top N v}{n^2} > \delta \right) \le \frac{4 \|N\|_F^2}{n^4 \delta^2}\;.
\end{equation}
\end{lemma}
\begin{proof}
To prove \eqref{eq:sym_quad_form}, observe that 
\begin{equation*}
    \E \left(\frac{u^\top M u}{n^2} \right)^2 = \E \left( \sum_{i\neq j} M_{ij} \frac{u_i u_j}{n^2}\right)^2 \stackrel{\text{(i)}}{=} 2 \E \left( \sum_{i\neq j} M_{ij}^2 \frac{u_i^2 u_j^2}{n^4}\right) \le \frac{8\|M\|_F^2}{n^4}\;.
\end{equation*}
where (i) follows from $u_i \indep u_j$ for any $i \neq j$. Then by applying Chebyshev's inequality, we obtain \eqref{eq:sym_quad_form}.
To prove \eqref{eq:asym_quad_form}, since $u_i \indep v_j$ for any $i, j$, we have
\begin{equation*}
\E \left(\frac{u^\top N v}{n^2} \right)^2 = \E \left( \sum_{i=1}^m \sum_{j=1}^k N_{ij} \frac{u_i v_j}{n^2}\right)^2 \stackrel{\text{(i)}}{=} \E \left( \sum_{i=1}^m \sum_{j=1}^k N_{ij}^2  \frac{u_i^2 v_j^2}{n^4}\right) \le \frac{4 \|N\|_F^2}{n^4}\;.
\end{equation*}
By Chebyshev's inequality, we obtain \eqref{eq:asym_quad_form}.
\end{proof}

The next Lemma gives the upper and lower bounds for the trace of projection matrix products, which is used in our analysis.
\begin{lemma}\label{lem:proj}
Let $M_1, M_2\in\R^{n\times n}$ be two projection matrices and $P$ be an $n\times n$ permutation matrix. We have
\begin{align}
    \tr(M_1 M_2) &\ge \max\{0, \tr(M_1) + \tr(M_2) - n\}\;, \label{eq:proj1}\\
    \tr(M_1 M_2 P) &\le \min\{\tr(M_1), \tr(M_2)\}\;. \label{eq:proj2}
\end{align}
\end{lemma}
\begin{proof}
Let $\mathrm{span}(M_1)$ and $\mathrm{span}(M_2)$ be the linear spaces spanned by projection matrices $M_1$ and $M_2$, and let $t = \mathrm{dim}(\mathrm{span}(M_1) \cap \mathrm{span}(M_2)), r = \mathrm{dim}(\mathrm{span}(M_1)), s = \mathrm{dim}(\mathrm{span}(M_2))$. Then, we express the projection matrices by $M_1 = U U^\top, M_2 = W W^\top$. Here, $U\in\R^{n\times r}$ and $W\in\R^{n\times s}$ are orthonormal matrices corresponding to $M_1$ and $M_2$, and $U$, $W$ have the same top $t$ columns from their intersection space. Then, $\tr(M_1 M_2) = \|U^\top W\|_F^2 \ge 0$. Note that $U^\top W$ is a $r\times s$ matrix, and one can easily verify that the top $t\times t$ block is the identity matrix by construction. Therefore, $\|U^\top W\|_F^2\ge t$. In addition, by the Grassmann's formula, 
\[
t = \mathrm{dim}(\mathrm{span}(M_1)) + \mathrm{dim}(\mathrm{span}(M_2)) - \mathrm{dim}(\mathrm{span}(M_1) \cup \mathrm{span}(M_2)) \ge r + s - n\;.
\]
For projection matrices, we have $r = \tr(M_1)$ and $s = \tr(M_2)$. Hence we obtain \eqref{eq:proj1}.

Secondly, let $\sigma_i(A)$ be the $i$-th singular value of a matrix $A$. By von Neumann’s trace inequality, for $n\times n$ matrices $A$ and $B$, we have $|\tr(AB)| \le \sum_{i=1}^n \sigma_{i}(A) \sigma_i(B)$. By setting $A = M_1 M_2$ and $B = P$, we have
\[
|\tr(M_1 M_2 P)| \le \sum_{i=1}^n \sigma_{i}(M_1 M_2) \sigma_i(P) = \sum_{i=1}^n \sigma_{i}(M_1 M_2)\;,
\]
where the last equality follows from the fact that $P$ is an orthogonal matrix and $\sigma_i(P) = 1$. Then, by the AM-GM inequality we have
\[
\sum_{i=1}^n \sigma_{i}(M_1 M_2) \le \sqrt{\mathrm{rank}(M_1 M_2)} \sqrt{\sum_{i=1}^n \sigma_{i}(M_1 M_2)^2} = \sqrt{\mathrm{rank}(M_1 M_2)} \|M_1 M_2\|_F\;.
\]
The basic rank inequality justifies that $\mathrm{rank}(M_1 M_2) \le \min\{\tr(M_1), \tr(M_2)\}$. Additionally, $\|M_1 M_2\|_F^2 = \tr(M_1 M_2) \le \min\{\tr(M_1), \tr(M_2)\}$. To summarize, we have $\tr(M_1 M_2 P) \le \min\{\tr(M_1), \tr(M_2)\}$ and the second result is proved.
\end{proof}

\end{document}